\documentclass{amsart}
\usepackage{amsmath}
\usepackage{amsthm}
\usepackage{amssymb,bbm}

\newtheorem{thm}{Theorem}[section]
\newtheorem{proposition}[thm]{Proposition}
\newtheorem{remark}[thm]{Remark}
\newtheorem{lemma}[thm]{Lemma}

\newtheorem{corollary}[thm]{Corollary}

\newcommand{\scalar}[2]{\langle{#1} \mspace{2mu}, {#2}\rangle}
\newcommand{\norm}[1]{\lVert #1 \rVert}
\newcommand{\nc}{\newcommand}

\newcommand{\R}{\mathbb{R}}
\newcommand{\C}{\mathbb{C}}
\newcommand{\RR}{\mathbb{R}}
\newcommand{\ZZ}{\mathbb{Z}}

\newcommand{\CC}{\mathbb{C}}

\def\11{\mathbbmss{1}}
\renewcommand{\i}{\mathrm{i}}
\newcommand{\e}{\mathrm{e}}
\renewcommand{\d}{\mathrm{d}}

\renewcommand{\Re}{\operatorname{Re}\,}
\renewcommand{\Im}{\operatorname{Im}\,}
\newcommand{\real}{\Re}
\newcommand{\imag}{\Im}

\nc{\ran}{\rangle}
\nc{\lan}{\langle}

\newcommand{\ls}{\lesssim}

\newcommand{\mc}[1]{\mathcal{#1}}                                             

\newcommand{\Sph}{\mathbb{S}}                                                 
\newcommand{\lapl}{\Delta}                                                    

\newcommand{\DETAILS}[1]{}

\allowdisplaybreaks

\begin{document}
\title[Hamiltonian Dynamics]{Hamiltonian dynamics of a particle interacting with a wave field}
\author[D. Egli, J. Fr\"ohlich, Z. Gang, A. Shao, I. M. Sigal]{Daniel Egli$^\dagger$, J\"urg Fr\"ohlich$^{\ddagger}$, Zhou Gang, Arick Shao$^\dagger$,\\ Israel Michael Sigal$^\dagger$}
\thanks{$\dagger$ The author's research was supported in part by NSERC under Grant No. NA7901}
\thanks{$\ddagger$ Present address: School of Mathematics, The Institute for Advanced Study, Princeton, NJ 08540, USA.  The author's visit at IAS is supported by `The Fund for Math' and `The Monell Foundation'}

\address{University of Toronto, Toronto, Canada}
\email{daniegli@math.utoronto.ca, ashao@math.utoronto.ca,\newline\indent\indent\indent\indent\indent\indent im.sigal@utoronto.ca}

\address{Institute for Theoretical Physics, ETH Zurich, CH-8093, Z\"urich, Switzerland}
\email{juerg@itp.phys.ethz.ch}

\address{University of Illinois at Urbana-Champaign, Urbana, USA}
\email{gangzhou@illinois.edu}

\date{}

\begin{abstract}
We study the Hamiltonian equations of motion of a heavy tracer particle interacting with a dense weakly interacting Bose-Einstein condensate in the classical (mean-field) limit.
Solutions describing ballistic subsonic motion of the particle through the condensate are constructed.
We establish asymptotic stability of ballistic subsonic motion.  
\end{abstract}

\maketitle

\section{Introduction} \label{sec.intro}

In this paper we study a physical system consisting of a very heavy quantum particle moving through a very dense, very weakly interacting Bose gas. The gas is at zero temperature, and it exhibits Bose-Einstein condensation. 
It is known that the speed of sound in an interacting Bose gas is strictly positive. (The dispersion law of sound waves, the Goldstone modes of the Bose-Einstein condensate, is linear in the magnitude of their wave vector at the origin in wave-vector space; see, e.g., \cite{bogol}.)
In a limiting regime where the density of the Bose gas is very large, but the interactions between atoms in the gas are very weak, the dynamics of the Bose gas is harmonic (the Hamiltonian is quadratic in the creation and annihilation operators of the Goldstone modes).

Quantum systems very similar to this one have been studied in \cite{JF,pizzo}.
Under appropriate assumptions on the interactions between the particle and the atoms in the Bose gas, one-particle states can be constructed that describe inertial motion of the dressed quantum particle at a speed below the speed of sound of the condensate.
Although there are some results on scattering in such systems, asymptotic stability or -completeness of this type of motion has not been established yet.
In this paper we study related problems in the classical or \emph{``mean-field''} limit.

One may argue quite convincingly (see, e.g., \cite{froehlich11}) that, in the mean-field limit where the density of the  Bose gas and the mass of the tracer particle become very large, the quantum dynamics of the system approaches the classical Hamiltonian dynamics of a point particle interacting with a wave medium (in the sense that time evolution and quantization commute, up to error terms that vanish in the mean-field limit).
The corresponding Hamiltonian equations of motion have been exhibited in \cite{froehlich11}. 
The purpose of the present paper is to construct a class of solutions of these equations describing \emph{inertial motion} of the particle and to establish \emph{asymptotic stability} of subsonic inertial motion for a suitable class of initial conditions.

In the following, $X_t \in\RR^{3}$ denotes the position of the tracer particle and $P_t \in \mathbb{R}^{3}$ denotes its momentum at time $t$.
The Bose-Einstein condensate is described, in the mean-field limit, by a complex wave field, $\beta_t$, belonging to a suitably chosen Sobolev space of functions on $\RR^{3}$. This wave field describes low-energy excitations of the ground state of the condensate.
In a certain limiting regime described in \cite{froehlich11}, and for a certain choice of boundary conditions of $\beta_t$ at $\infty$, the Hamiltonian equations of motion of the system are given by
\begin{align} 
\label{eq.emodel_full} M \dot{X}_t &= P_t \text{,} \\
\notag \dot{P}_t &= - \nabla V (X_t) +  \int_{\R^3} \nabla W( \cdot - X_t) \Re \beta_t \text{,} \\
\notag \i \dot{\beta}_t &= - \Delta \beta_t +  \mu\Re \beta_t + W( \cdot - X_t) \text{.}
\end{align}
Here, $V: \R^3 \rightarrow \R$ is the potential of an external force acting on the tracer particle, and $W: \R^3 \rightarrow \R$ is a two-body potential describing interactions between the tracer particle and the atoms in the Bose gas.
Furthermore, $M > 0$ denotes the mass of the tracer particle, and $\mu$ is a parameter related to the product of the density of the condensate with the strength of interactions among atoms in the Bose gas, which determines its speed of sound, $\sqrt{\mu}$.

One expects the particle motion to exhibit qualitatively different properties according to whether the initial speed of the particle is \emph{subsonic} or \emph{supersonic}.
For a supersonic initial velocity, one expects the particle to decelerate through emission of Cerenkov radiation of sound waves into the Bose-Einstein condensate---its speed asymptotically approaching the speed of sound in the condensate.
For the time being, this expectation remains a conjecture that will be analyzed elsewhere.
In this paper, we construct solutions of the equations of motion describing inertial motion of the tracer particle accompanied by a splash in the condensate moving at the same velocity.
The main result of this paper states that this type of motion is asymptotically stable.
More precisely, we show that the system \eqref{eq.emodel_full} of equations has ``traveling wave solutions",
\begin{align}\label{eq:inertial}
X_t = vt \text{,} \qquad \quad P_t = M v \text{,} \qquad \beta_t (x) = \gamma_v (x - vt) \text{,}
\end{align}
describing inertial motion if the speed of the particle, $|v|$, is smaller than or equal to the speed of sound in the condensate, and we show that these solutions are stable under small perturbations.

Simplified statements of our main results are summarized in the following theorems; see Section \ref{sec.trwave} for the full precise statements.
The outlines for their proofs can be found in Sections \ref{sec.inertpf_intro} and \ref{sec.stabpf_intro}.

\begin{thm}[Inertial motion] \label{thm.traveling_wave_reg0}
Consider the system \eqref{eq.emodel_full}, with $W$ smooth and rapidly decreasing, and with $V \equiv 0$.
Fix also a velocity vector $v \in \R^3$.
\begin{itemize}
\item If $| v | < \sqrt{\mu}$, then \eqref{eq.emodel_full} has a unique smooth solution of the form \eqref{eq:inertial}, with
\[ | \Re \gamma_v (x) | \lesssim \lan x \ran^{(-3)+} \text{,} \qquad | \Im \gamma_v (x) | \ls \lan x \ran^{(-2)+} \text{.} \]

\item If $| v | = \sqrt{\mu}$, then \eqref{eq.emodel_full} has a unique solution of the form \eqref{eq:inertial}, with
\[ | \Re \gamma_v (x) | \ls \lan x \ran^{(-1)+} \text{,} \qquad | \Im \gamma_v  (x)| \ls \lan x \ran^{(-\frac{2}{3})+} \text{.} \]

\item If $| v | > \sqrt{\mu}$, then \eqref{eq.emodel_full} has no finite-energy traveling wave solutions, except for the non-generic case that the Fourier transform of $W$ vanishes on
\[ \{ \xi \in \RR^3 \mid |\xi|^2 + \mu - ( v \cdot \hat{\xi} )^2 = 0 \} \text{.} \]
\end{itemize}
\end{thm}

\begin{thm}[Stability of subsonic inertial motion] \label{thm.technical0}
Assume that $W$ is spherically symmetric, smooth, and rapidly decreasing, and that its Fourier transform is nowhere vanishing.
Suppose $(X_t, P_t, \beta_t)$ is a solution to the system \eqref{eq.emodel_full} in the case $V \equiv 0$.
If the initial conditions satisfy
\begin{itemize}
\item $| v_0 | := M^{-1} | P_0 | < \sqrt{\mu}$, and
\item $\beta_0 - \gamma_{v_0}$ is smooth and rapidly decreasing, and $\| \langle x \rangle^4 (\beta_0 - \gamma_{v_0}) \|_{ L^2 }$ is sufficiently small (depending on $M$, $\mu$, $|v_0|$ and $W$),
\end{itemize}
then there exists an asymptotic velocity $v_\infty \in \R^3$, with $| v_\infty | < \sqrt{\mu}$, such that
\[ | M^{-1} P_t - v_\infty | \lesssim (1+t)^{-3} \text{.} \]
Moreover, there exists an asymptotic initial position $\bar{X}_0 \in \mathbb{R}^3$ such that
\begin{align*}
| X_t - v_\infty t - \bar{X}_0 | &\lesssim (1 + t)^{-2} \text{,} \\
\| \Re [ \beta_t - \gamma_{v_\infty} (\cdot - \bar{X}_0 - v_\infty t) ] \|_{ L^\infty } &\lesssim (1 + t)^{-\frac{3}{2}} \text{,} \\
\| \Im [ \beta_t - \gamma_{v_\infty} (\cdot - \bar{X}_0 - v_\infty t) ] \|_{ L^\infty } &\lesssim (1 + t)^{(-1)+} \text{.}
\end{align*}
\end{thm}

In other words, the particle velocity converges to some asymptotic subsonic value as $t \nearrow \infty$.
Moreover, the trajectory of the particle also converges to an asymptotic ballistic path, $t \mapsto \bar{X}_0 + v_\infty t$, and the field $\beta$ converges to the traveling wave solution associated with this ballistic motion.

In \cite{froehlich12} and in \cite{egli11} (where a more general result going beyond the weak-coupling limit is proven), some of the authors have analyzed particle motion in an ideal Bose gas, that is, for $\mu=0$.
In an ideal Bose gas, the speed of sound is zero, hence there is only supersonic particle motion.
For the models of \cite{froehlich12} and \cite{egli11}, the heuristic picture sketched above is vindicated.
The particle decelerates to zero speed (= the speed of sound), $|P_t|=O(t^{-\alpha})$, for some rate $\alpha>\frac{1}{2}$ possibly depending on initial conditions.
The condensate wave $\beta_t$ is proven to approach the static solution of equation \eqref{eq.emodel_full}.

The existence of solutions describing inertial motion at non-zero velocity shows that the static solutions are unstable if $\mu>0$.
However, for an ideal Bose gas, that is, $\mu=0$, the static solutions are stable (\cite{froehlich12, froehlich11}).

A variety of physics-style results for related problems can be found in \cite{kovrizhin01, ovchinnikov98}.
In \cite{bdb}, the authors prove asymptotic ballistic motion for a forced particle interacting with a wave field.
However, in their work they assume an \textit{independent} wave field for each point in space.
There is extensive mathematical literature on existence, nonexistence, stability, and scattering of traveling waves for nonlinear Schr\"odinger and reaction-diffusion equations.
For the latter, see for instance \cite{vvv}, or \cite{abvv}, for recent work.
For scattering theory in nonlinear Schr\"odinger equations, see \cite{sw1,sw2,bp,fty, gustafson06, gustafson09}.
For existence and nonexistence results, see \cite{bethuel99,g03,m08}.
For results on stability of traveling waves, see \cite{gss,bs,c01,rss,bgss08,gz08}.
See also \cite{m10} for a recent review of results on traveling waves in nonlinear Schr\"odinger equations with nonzero conditions at infinity.

\subsection{Notations} \label{sec.intro_not}

Here we summarize some notation used throughout the paper.
\begin{itemize}
\item $f\lesssim_k g$ means that there is a constant $C$ depending on $k$ such that $f\leq Cg$.

\item We use the standard notation
\[ \langle x \rangle = ( 1 + | x |^2 )^\frac{1}{2} \text{,} \qquad \hat{x} = \frac{x}{ |x| } \in \Sph^{n-1} \text{,} \qquad x \in \R^n \text{.} \]

\item Given $f: \R^n \rightarrow \C$ and $X \in \R^n$, we define the spatial translate $f^X$ of $f$ by
\[ f^X (x):= f (x-X) \text{.} \]

\item Given $f: \R^n \rightarrow \C$, we let $\mc{F} (f)$ and $\hat{f}$ denote the Fourier transform of $f$.

\item We will also use $\langle \cdot, \cdot \rangle$ to denote the standard $L^2$-inner product on $\R^n$,
\[ \langle f, g \rangle = \int_{ \R^n } f \bar{g} \text{,} \qquad f, g : \R^n \rightarrow \C \text{.} \]
\end{itemize}
In most of this paper, we will only consider three-dimensional space, $n = 3$.
In a few instances, we will also use the same conventions in the case $n = 1$.

Next, we will use the following notation.
\begin{itemize}
\item Let $\mc{S} (\R^n)$ denote the space of smooth rapidly decreasing functions on $\R^n$.

\item For any $b \in \R$, we let $H^b (\R^n)$ denote the standard (inhomogeneous) fractional Sobolev function space, defined using the norm
\[ \| f \|_{ H^b }^2 = \int_{\R^n} \langle \xi \rangle^{2b} | \hat{f} (\xi) |^2 d \xi \text{.} \]

\item Furthermore, when $b < 3/2$, we can define the corresponding homogeneous Sobolev space, $\dot{H}^b (\R^n)$, via the norm
\footnote{Such spaces can still be adequately defined when $b \geq 3/2$, but only as tempered distributions modulo polynomials.  For simplicity, we avoid such technicalities here.}
\[ \| f \|_{ \dot{H}^b }^2 = \int_{\R^n} | \xi |^{2b} | \hat{f} (\xi) |^2 d \xi \text{.} \]
\end{itemize}

In our analysis, we will often encounter functions with different regularities at low and high frequencies.
These functions will typically belong to the generalized Sobolev spaces $H_a^b (\R^n)$, which treat low frequencies like $\dot{H}^a$ and high frequencies like $\dot{H}^b$.
Again, we will restrict our attention to spaces with $a < 3/2$. The spaces $H_a^b (\R^n)$ can then be described in terms of the norm
\[ \| f \|_{ H_a^b (\R^n) }^2 = \int_{ |\xi| \leq 1 } | \xi |^{2a} | \hat{f} (\xi) |^2 d \xi + \int_{ |\xi| \geq 1 } | \xi |^{2b} | \hat{f} (\xi) |^2 d \xi \text{.} \]
The spaces $H^b$ and $\dot{H}^b$ represent special cases, with $a = 0$ and $a = b$, respectively.

The following simple observations relate these $H_a^b$-spaces, which are somewhat technical in nature, to some more intuitive situations.
\begin{itemize}
\item If $f \in L^1 (\R^n) \cap H^b (\R^n)$, then $f \in H_a^b (\R^n)$ for any $a > -n/2$.

\item If $f \in \mc{S} (\R^n)$, then $f \in H_a^b (\R^n)$ for any $a > -n/2$ and $b \in \R$.
\end{itemize}

Finally, we are interested in decay properties of fields and particle.
In general, there are two obstacles preventing decay:
\begin{itemize}
\item Singular behavior of various operators near the Fourier origin---these are properties of our system that limit the optimal rates of decay.

\item Lack of spatial decay of our initial conditions and potential.
\end{itemize}
We choose our assumptions such that the decay will only be limited by the first obstacle mentioned above.
For this purpose, we introduce the following definition.
We say a function $f: \R^n \rightarrow \C$ is \emph{sufficiently fast decaying} iff
\begin{equation} \label{eq.gz_norm} \| \langle x \rangle^N f \|_{ L^2 }^2 = \int_{ \R^n } \langle x \rangle^{2N} | f (x) |^2 dx < \infty \text{,} \end{equation}
for some large enough $N$.

\begin{remark}
In this paper, taking $N = 6$ is sufficient to prove all results.
\end{remark}

\section{Basic properties of the equations of motion.} \label{sec.equations}

In this section we discuss some basic properties of the system of equations \eqref{eq.emodel_full}.
Moreover, from now on, we will only consider the case $V \equiv 0$, in which there are no external forces acting on the tracer particle.

\subsection{Hamiltonian Structure and Transformations} \label{sec.equations_sym}

The system \eqref{eq.emodel_full} is Hamiltonian, with Hamilton functional
\begin{align}
\label{eq.energy} \mc{H} (X, P, \beta) &= \frac{1}{2 M} | P |^2 + \int_{\R^3} \left[ \frac{1}{2} |\nabla \beta |^2 + \mu |\Re \beta|^2 + W^X \Re \beta \right] \text{,}
\end{align}
and with the symplectic form
\[ \omega ( (X, P, \beta), (X^\prime, P^\prime, \beta^\prime) ) = P' \cdot X - P \cdot X' +  \Im \int_{\R^3} \bar{\beta} \beta^{\prime} \text{.} \]
In the preceding expressions, we have, for conciseness, adopted the usual practice of identifying vector spaces with their tangent and cotangent spaces.

Moreover, the system has the following invariance and covariance properties:
\begin{enumerate}
\item[(a)] \emph{Time translation}: If $(X, P, \beta)$ is a solution of \eqref{eq.emodel_full}, and if $t_0 \in \R$, then
\[ \tilde{X}_t = X_{t - t_0} \text{,} \qquad \tilde{P}_t = P_{t - t_0} \text{,} \qquad \tilde{\beta}_t (x) = \beta_{t - t_0} (x) \]
is a solution of \eqref{eq.emodel_full}.

\item[(b)] \emph{Spatial translation}: If $(X, P, \beta)$ is a solution of \eqref{eq.emodel_full}, and if $x_0 \in \R^3$, then
\[ \tilde{X}_t = x_0 + X_t \text{,} \qquad \tilde{P}_t = P_t \text{,} \qquad \tilde{\beta}_t (x) = \beta_t (x - x_0) \]
is a solution of \eqref{eq.emodel_full}.
In other words, \eqref{eq.emodel_full} has translation symmetry.

\item[(c)] \emph{Spatial rotation}: If $(X, P, \beta)$ is a solution of \eqref{eq.emodel_full}, and if $R \in SO (3)$, i.e., $R$ is a spatial rotation, then
\[ \tilde{X}_t = R^{-1} X_t \text{,} \qquad \tilde{P}_t = R^{-1} P_t \text{,} \qquad \tilde{\beta}_t (x) = \beta_t ( R x ) \]
is another solution of \eqref{eq.emodel_full}, with $W$ replaced by $\tilde{W} (x) = W (Rx)$.
In particular, if $W$ is spherically symmetric, then \eqref{eq.emodel_full} is covariant under rotations.

\item[(d)] If $(X, P, \beta)$ is a solution of \eqref{eq.emodel_full}, and if $\lambda > 0$, then
\[ \tilde{X}_t = X_t \text{,} \qquad \tilde{P}_t = \lambda^{-1} P_t \text{,} \qquad \tilde{\beta}_t (x) = \lambda^{-\frac{1}{2}} \beta_t ( x ) \]
is a solution of \eqref{eq.emodel_full}, with $M$ and $W$ replaced by
\[ \tilde{M} = \lambda^{-1} M \text{,} \qquad \tilde{W} (x) = \lambda^{-\frac{1}{2}} W ( x ) \text{.} \]

\item[(e)] If $(X, P, \beta)$ is a solution of \eqref{eq.emodel_full}, and if $\lambda > 0$, then
\[ \tilde{X}_t = \lambda X_{ \lambda^{-2} t } \text{,} \qquad \tilde{P}_t = \lambda^{-1} P_{ \lambda^{-2} t } \text{,} \qquad \tilde{\beta}_t (x) = \lambda^{-\frac{3}{2}} \beta_{ \lambda^{-2} t } ( \lambda^{-1} x ) \]
is a solution of \eqref{eq.emodel_full}, with $\mu$ and $W$ replaced by
\[ \tilde{\mu} = \lambda^{-2} \mu \text{,} \qquad \tilde{W} (x) = \lambda^{-\frac{7}{2}} W ( \lambda^{-1} x ) \text{.} \]
\end{enumerate}
The last two properties can be interpreted as transformations of \eqref{eq.emodel_full} corresponding to changes of units.
All these properties can be checked by direct computation.

By rotation and translation transformations ((b) and (c)), we may assume, where convenient, that $X_0 = 0$ and $P_0 = (0, 0, |P_0|)$.
Moreover, by appropriate changes of units ((d) and (e)), we need only consider the specific case $M = 1$ and $\mu = 1$.
In this setting, the equations of motion are
\begin{align} 
\label{eq.emodel} \dot{X}_t &= P_t \text{,} \\
\notag \dot{P}_t &=  \int_{\R^3} \nabla W^{X_t} \Re \beta_t \d x \text{,} \\
\notag \i \dot{\beta}_t &= - \Delta \beta_t + \Re \beta_t + W^{X_t} \text{.}
\end{align}
\emph{From now on, we exclusively analyze equations \eqref{eq.emodel} rather than \eqref{eq.emodel_full}.}

\begin{remark}
As the mass $M$ of the tracer particle has been set to $1$, the velocity $v = \dot{X}$ and the momentum $P$ coincide numerically.
Depending on the appropriate physical interpretation, we will use the notations $v_t$ and $P_t$, respectively.
\end{remark}

The symmetries of the system \eqref{eq.emodel} yield standard conserved quantities, whenever such quantities are well-defined:
\begin{itemize}
\item By time translation symmetry, solutions of \eqref{eq.emodel} have the energy $\mc{H}$ defined in \eqref{eq.energy} (with $M = \mu = 1$) as a conserved quantity.

\item By space translation symmetry, solutions of \eqref{eq.emodel} have conserved momentum
\begin{equation} \label{eq.momentum} \mc{P}_j ( X, P, \beta ) = P_j + \frac{1}{2} \Im \int_{ \R^3 } \bar{\beta} \partial_j \beta \text{,} \qquad 1 \leq j \leq 3 \text{.} \end{equation}

\item If $W$ is spherically symmetric, then the rotational symmetry property implies that solutions of \eqref{eq.emodel} have conserved angular momentum.
\end{itemize}

\subsection{Global Well-Posedness} \label{sec.equations_gwp}

A basic question is whether solutions to the system \eqref{eq.emodel} exist and are unique, for ``general" initial data.
The following theorem gives an affirmative answer to this question.

\begin{thm} \label{thm.emodel_gwp}
Let $a \in (-5/2, 1/2)$, $b \in [1, \infty)$, and $W \in L^1 (\R^3) \cap H^b (\R^3)$.
If
\[ P_0 \in \R^3 \text{,} \qquad \real \beta_0 \in H_a^b (\R^3) \text{,} \qquad \imag \beta_0 \in H_{a+1}^b (\R^3) \text{,} \]
then the system \eqref{eq.emodel} has a unique solution
\[ X_t, P_t \in C ( [0, \infty); \R^3 ) \text{,} \qquad (\real \beta_t, \imag \beta_t) \in C ( [0, \infty); H_a^b (\R^3) \times H_{a+1}^b (\R^3) ) \text{,} \]
corresponding to the initial data $X_0 $, $P_0$, and $\beta_0$.
\end{thm}

The proof of Theorem \ref{thm.emodel_gwp}, which employs rather conventional contraction mapping techniques, is deferred until Appendix \ref{sec.gwp}.

\begin{remark}
Note that if
\[ \real \beta_0 \in H^1 (\R^3) \text{,} \qquad \imag \beta_0 \in \dot{H}^1 (\R^3) \text{,} \qquad W \in L^2 (\R^3) \text{,} \]
then the Hamiltonian functional $\mc{H}$ in \eqref{eq.energy} is well-defined.
This case is covered by Theorem \ref{thm.emodel_gwp}, with parameters $a = 0$ and $b = 1$.
\end{remark}

While Theorem \ref{thm.emodel_gwp} shows that $P$ and $\beta$ cannot blow up in finite time, it does not exclude the possibility that $P$ and $\beta$ blow up in the limit $t \nearrow \infty$.
However, if the Hamiltonian functional is well-defined, then it must be conserved for solutions of \ref{thm.emodel_gwp}.
In this case, one is able to not only rule out blowup, but actually establish uniform bounds on both quantities.

\begin{proposition} \label{thm.emodel_energy}
Let $W \in L^2 (\R^3)$, and suppose that
\[ X_t, P_t \in C ( [0, \infty); \R^3 ) \text{,} \qquad (\real \beta_t, \imag \beta_t) \in C ( [0, \infty); H^1 (\R^3) \times \dot{H}^1 (\R^3) ) \]
is a solution of \eqref{eq.emodel}.
Then, for any $t \in [0, \infty)$,
\begin{align}
\label{eq.emodel_energy} \frac{1}{2} | P_t |^2 + \int_{\R^d} \left( \frac{1}{2} | \nabla\beta_t |^2 + \frac{1}{2} | \real \beta_t |^2 \right) dx \leq \mc{H} ( X_0, P_0, \beta_0 ) + \frac{1}{2}\| W \|_{ L^2_x }^2 \text{.}
\end{align}
\end{proposition}

\begin{proof}
By conservation of the Hamilton functional, 
\begin{align}
\mc{H} ( X_t, P_t, \beta_t ) = \mc{H} ( X_0, P_0, \beta_0 ) = \mc{H}_0 \text{.}
\end{align}
It then follows from the definition of $\mc{H}$ that
\begin{align*}
\frac{1}{2} | P_t |^2 + \int_{\R^d} \left( \frac{1}{2}|\nabla\beta_t |^2 + | \real \beta_t |^2 \right) dx &= \mc{H}_0 - \int_{\R^d} W^{X_t} \real \beta_t dx \\
&\leq \mc{H}_0 + \| W \|_{ L^2 } \| \real \beta_t \|_{ L^2 } \text{.}
\end{align*}
The inequality \eqref{eq.emodel_energy} follows, since
\[ \| W \|_{ L^2 } \| \real \beta_t \|_{ L^2 } \leq \frac{1}{2} \| W \|_{ L^2 }^2 + \frac{1}{2} \| \real \beta_t \|_{ L^2 }^2 \text{.} \qedhere \]
\end{proof}

\subsection{A Change of Variables} \label{sec.equations_cov}

From a mathematical perspective, the form \eqref{eq.emodel} of the equation of motion for $\beta$ is inconvenient because the operator $(-\Delta + \Re)$ is not complex linear.
As a result, $\Re \beta$ and $\Im \beta$ will have inherently different regularity properties in low frequencies.
To address this, we make a change of variables, introduced in \cite{gustafson06} for the Gross-Pitaevskii equation.
In the new variables, the field equation gives rise to an evolution that is unitary (on the $H_a^b$-spaces).

First, let
\[ U = ( - \lapl )^\frac{1}{2} ( 1 - \lapl )^{-\frac{1}{2}} \text{,} \]
and define the $\R$-linear operator $U_r: \mc{S} (\R^3) \rightarrow \mc{S} (\R^3)$ by
\begin{equation} \label{eq.op_V} U_r u = U ( \Re u ) + \i ( \Im u ) \text{,} \end{equation}
extended appropriately to various Sobolev-type spaces of functions.
Note that both $U$ and $U_r$ are invertible on the appropriate spaces.

Next, we set
\begin{align*}
B_t := U_r^{-1} \beta_t = U^{-1} ( \Re \beta_t ) + \i ( \Im \beta_t ) \text{.}
\end{align*}
Direct computation shows that, by letting
\[ H = \sqrt{ - \lapl ( - \lapl + 1) } \text{,} \]
the equations of motion in terms of $B_t$ take the form
\begin{equation} \label{eq.emodel_alt} \dot{X}_t = P_t \text{,} \qquad \dot{P}_t = \int_{ \R^3 } \nabla W^{X_t} U \Re B_t \text{,} \qquad \i \dot{B}_t = H B_t + W^{X_t} \text{.} \end{equation}
Duhamel's principle yields the following expression for $B_t$:
\begin{equation} \label{eq.B_duhamel} B_t = \e^{- \i H t} B_0 - \i \int_0^t \e^{ -\i H (t - s) } W^{X_s}\d s \text{.} \end{equation}
This representation will be advantageous, as it highlights the dispersive properties of the system.

From now on, we use the notation
\begin{equation} \label{eq.Hv} H_v := H + \i v \cdot \nabla \text{,} \qquad v \in \R^3 \text{,} \end{equation}
as this operator will be utilized repeatedly in our analysis.
By Fourier transformation, we see that the symbol of $H_v$ is given by
\begin{equation} \label{eq.hv} h_v (\xi) := | \xi | \langle \xi \rangle - v \cdot \xi = | \xi | ( \langle \xi \rangle - v \cdot \hat{\xi} ) \text{.} \end{equation}

\subsection{Symbol Classes} \label{sec.symbol}

In subsequent decay estimates, we will need to consider various differential operators constructed from $\nabla$, $U$, $H_v$, etc.; see Section \ref{sec.equations_cov}.
In particular, we require specific estimates for derivatives of the symbols of these operators.
Of special importance is how these symbols behave near the Fourier origin.
As we will do this repeatedly for several operators, it becomes convenient to treat them in a systematic manner, depending on their behavior in large and small frequencies.
As a result, we define the following classes of functions.

First, let $a, b \in \RR$, and suppose $f: \R^n \setminus \{ 0 \} \rightarrow \C$ is smooth.
We say that $f$ is \emph{in the class $\mc{M}_a^b$} iff for every multi-index $\alpha$, we have
\begin{equation} \label{eq.MN} | \partial^\alpha f (\xi) | \lesssim_{f, \alpha} \begin{cases} | \xi |^{b - |\alpha|} & | \xi | \geq 1 \text{,} \\ | \xi |^{a - |\alpha|} & | \xi | \leq 1 \text{,} \end{cases} \qquad \xi \in \R^n \text{.} \end{equation}
These resemble the usual H\"ormander classes encountered in pseudodifferential analysis, except, like for the $H_a^b$-spaces, we treat low and high frequencies separately.
Note the following observations:
\begin{itemize}
\item The map $\xi \mapsto | \xi |^a \langle \xi \rangle^{b - a}$, $\xi \in \R^n \setminus \{ 0 \}$, is in the class $\mc{M}_a^b$.

\item Given $a_1, a_2, b_1, b_2 \in \RR$, the product of two functions in the classes $\mc{M}_{a_1}^{b_1}$ and $\mc{M}_{a_2}^{b_2}$, respectively, is a function in the class $\mc{M}_{a_1 + a_2}^{b_1 + b_2}$.
\end{itemize}
The above can be verified through straightforward computations.

Next, we consider similar behavior in only the radial direction.
Given a smooth function $g: (0, \infty) \rightarrow \R$, we say that $g$ is \emph{in the class $\mc{O}^b$} iff for any integer $N \geq 0$,
\begin{equation} \label{eq.ON} | g^{ (N) } (\rho) | \lesssim_{g, N} \langle \rho \rangle^{b - N} \text{,} \qquad \rho \in (0, \infty) \text{.} \end{equation}
Unlike for the classes $\mc{M}_a^b$, here we will need to track the powers of $\rho$ explicitly, hence we only treat the powers of $\langle \rho \rangle$ in a unified way.
Note the following:
\begin{itemize}
\item The map $\rho \mapsto \langle \rho \rangle^b$, $\rho \in (0, \infty)$, is of the form $\mc{O}^b$.

\item Given $b_1, b_2 \in \RR$, the product of two functions in the class $\mc{O}^{b_1}$ and $\mc{O}^{b_2}$, respectively, is a function in the class $\mc{O}^{b_1 + b_2}$.
\end{itemize}
Again, these properties are verified easily.

Now, for any $\theta \in \Sph^{n-1}$, we can define the restriction
\[ f_\theta: (0, \infty) \rightarrow \C \text{,} \qquad f_\theta (\rho) = f (\rho \theta) \text{.} \]
In other words, we represent the arguments of $f$ using polar coordinates, and we restrict $f$ to a single spherical parameter.
We say that $f$ itself is \emph{in the class $\mc{O}^b$} iff each $f_\theta$, $\theta \in \Sph^{n-1}$, is in the class $\mc{O}^b$.

Note the following properties for some basic operators:
\begin{itemize}
\item The symbols $\xi$ and $|\xi|$ for $\nabla$ and $|\nabla|$, respectively, are in the class $\mc{M}^1_1$.

\item Given any $b \in \R$, the operator $U^b$ (see Section \ref{sec.equations_cov}) has associated symbol $u^b (\xi) = | \xi |^b \langle \xi \rangle^{-b}$, which is in the class $\mc{M}_b^0$.
\end{itemize}
Of particular interest are the operators $H_v$, $v \in \R^3$, defined in \eqref{eq.Hv}.
The associated symbol, $h_v$, was defined in \eqref{eq.hv}.
In particular, we have
\[ h_{v, \theta}: (0, \infty) \rightarrow \R \text{,} \qquad h_{v, \theta} (\rho) = \rho ( \langle \rho \rangle - v \cdot \theta ) \text{,} \]
for each $\theta \in \Sph^2$.
Note that when $|v| < 1$, we have the uniform lower bound
\[ | \xi |^{-1} | h_v (\xi) | = | \langle \xi \rangle - v \cdot \hat{\xi} | \geq 1 - | v | > 0 \text{,} \qquad \xi \in \R^3 \setminus \{ 0 \} \text{.} \]
Some other properties of $h_v$, listed below, can also be computed directly:
\begin{itemize}
\item If $b \in [0, \infty)$, then $(h_v)^b$ is in the class $\mc{M}_b^{2b}$, with the constants in \eqref{eq.MN} depending on $|v|$ and $b$.
Moreover, if $|v| < 1$, then $(h_v)^{-b}$ is well-defined on $\R^3 \setminus \{ 0 \}$ and is in the class $\mc{M}_{-b}^{-2b}$.

\item If $| v | < 1$, then for any $\theta \in \Sph^2$,
\begin{equation} \label{eq.phase_subsonic} h_{v, \theta}^\prime (\rho) \simeq_{ |v| } \langle \rho \rangle \text{.} \end{equation}

\item If $v \in \R^3$ and $\theta \in \Sph^2$, then $h_{v, \theta}^{\prime\prime}$ is in the class $\mc{O}^0$.
Moreover, the associated constants in \eqref{eq.ON} are independent of $v$.
\end{itemize}

\section{The Main Results} \label{sec.trwave}

In this section, we state the main results of this paper.
These pertain to traveling wave solutions of \eqref{eq.emodel} of the form
\begin{equation} \label{eq:trawave} X_t = vt \text{,} \qquad P_t \equiv v \text{,} \qquad \beta_t (x) = \gamma_v (x - vt) \text{,} \end{equation}
where the function $\gamma_v$ is time-independent.
These solutions describe uniform motion of the particle, along with a splash, $\gamma_v$, of the Bose gas accompanying the particle.

The main results of this section, and of this paper, are the following:
\begin{enumerate}
\item The existence of subsonic and sonic inertial solutions, and the generic nonexistence of finite energy supersonic inertial solutions (Theorem \ref{thm.traveling_wave_reg}).

\item The stability of subsonic inertial particle motion (Theorem \ref{thm:technical}).
\end{enumerate}

\subsection{Inertial Solutions} \label{sec.trwave_free}

The first main task is to show that traveling wave solutions of the form \eqref{eq:trawave} exist and are unique.
Furthermore, we obtain spatial decay estimates for these solutions.
The precise statements are given below.

\begin{thm} \label{thm.traveling_wave_reg}
Let $v \in \R^3$, and let $b \in [-1/2, \infty)$.
\begin{itemize}
\item If $| v | < 1$, then given any $W \in L^1 (\R^3) \cap H^b (\R^3)$, there exists a unique traveling wave solution \eqref{eq:trawave} to \eqref{eq.emodel}, with
\begin{equation} \label{eq.traveling_wave_subsonic} \Re \gamma_v \in \bigcap_{a > -\frac{3}{2}} H_a^{b+2} (\R^3) \text{,} \qquad \Im \gamma_v \in \bigcap_{a > -\frac{1}{2}} H_a^{b+2} (\R^3) \text{.} \end{equation}
Moreover, if $W$ is sufficiently fast decaying (see \eqref{eq.gz_norm}), then for any $\epsilon > 0$,
\begin{equation} \label{eq.traveling_wave_subsonic_decay} | \Re \gamma_v (x) | \lesssim_\epsilon \langle x \rangle^{-3 + \epsilon} \text{,} \qquad | \Im \gamma_v (x) | \lesssim_\epsilon \langle x \rangle^{-2 + \epsilon} \text{.} \end{equation}

\item If $| v | = 1$, then for any $W \in L^1 (\R^n) \cap H^b (\R^3)$, there exists a unique traveling wave solution \eqref{eq:trawave} to \eqref{eq.emodel}, with
\begin{equation} \label{eq.traveling_wave_sonic} \Re \gamma_v \in \bigcap_{a > -\frac{1}{2}} H_a^{b+2} (\R^3) \text{,} \qquad \Im \gamma_v \in \bigcap_{a > \frac{1}{2}} H_a^{b+2} (\R^3) \text{,} \end{equation}
Moreover, if $W$ is sufficiently fast decaying, then for any $\epsilon > 0$,
\begin{equation} \label{eq.traveling_wave_sonic_decay} | \Re \gamma_v (x) | \lesssim_\epsilon \langle x \rangle^{ -1 + \epsilon } \text{,} \qquad | \Im \gamma_v (x) | \lesssim_\epsilon \langle x \rangle^{ -\frac{2}{3} + \epsilon } \text{.} \end{equation}

\item If $|v| > 1$, then for any $W \in \mc{S} (\RR^3)$ such that $\hat{W}$ does not vanish on the set
\[ \{ \xi \in \RR^3 \mid 1 + | \xi |^2 - ( v \cdot \hat{\xi} )^2 = 0 \} \text{,} \]
there do not exist any finite-energy solutions \eqref{eq:trawave} to \eqref{eq.emodel}.
\end{itemize}
\end{thm}

The proof of Theorem \ref{thm.traveling_wave_reg} is given in Section \ref{sec.inertpf}.

\begin{remark}
Theorem \ref{thm.traveling_wave_reg} implies the following behavior of the profile $\gamma_v$:
\begin{itemize}
\item If $|v| < 1$, both $\real \gamma_v$ and $\imag \gamma_v$ belong to $L^2 (\R^3)$.

\item If $|v| = 1$, only $\real \gamma_v$ must belong to $L^2 (\R^3)$, though $\imag \gamma_v$ is in $\dot{H}^1 (\R^3)$.

\item If $|v| > 1$, neither $\real \gamma_v$ nor $\imag \gamma_v$ must belong to $L^2 (\R^3)$.
\end{itemize}
\end{remark}

A special case of traveling wave solutions \eqref{eq:trawave} is the case $v = 0$, that is, the stationary solutions of \eqref{eq.emodel}.
In this case, $\gamma_0$ is real-valued and satisfies
\[ \gamma_0 = - ( \Delta + 1 )^{-1} W \text{.} \]
This expression implies that $\gamma_0 (x)$ decays at infinity like either $\e^{-|x|}$ or $W$, whichever decays at the slower rate.

\subsection{Stability of Subsonic Inertial Solutions} \label{stabtrav}

Our main aim in this paper is to study the stability of the inertial solutions \eqref{eq:trawave}, in the subsonic case ($|v| < 1$).
Roughly speaking, we prove that if the initial speed of the particle is subsonic and the initial field is close to a corresponding traveling-wave solution then the velocity of the particle converges to an asymptotic (limiting) subsonic velocity and the field converges uniformly to a traveling wave corresponding to that limiting velocity.

A precise statement of our asymptotic stability result is given below. This result implies Theorem \ref{thm.technical0}.

\begin{thm} \label{thm:technical}
Suppose $W \in H^1 (\R^3)$ is sufficiently fast decaying and spherically symmetric, and let
\[ X_0 = 0 \text{,} \qquad P_0 = v_0 \in \R^3 \text{,} \qquad \beta_0 \in H^1 (\R^3) \text{,} \]
be initial data for the system \eqref{eq.emodel}, such that:
\begin{itemize}
\item $|v_0| < 1$, i.e., the particle is initially subsonic.

\item $\hat{W}$ is nowhere vanishing.

\item There is some $\varepsilon_0 > 0$, depending on $| v_0 |$ and $W$, such that 
\[ \| \langle x \rangle^4 (\beta_0 - \gamma_{v_0}) \|_{ L^2 } < \varepsilon_0 \text{.} \]
\end{itemize}
Then, there is a global solution $(X, P = v, \beta)$ to the system \eqref{eq.emodel} with the above initial data, in the sense of Theorem \ref{thm.emodel_gwp}.
Furthermore:
\begin{itemize}
\item There exist an asymptotic velocity $v_\infty \in \R^3$, with $| v_\infty | < 1$, and an asymptotic initial position $\bar{X}_0 \in \R^3$ such that
\begin{align}
\label{eq:vtconv} | v_t - v_\infty | &\lesssim \varepsilon_0 (1 + t)^{-3} \text{,} \\
\label{eq:xtconv} | X_t - v_\infty t - \bar{X}_0 | &\lesssim \varepsilon_0 (1 + t)^{-2} \text{.}
\end{align}
\item If $\nabla (\beta_0-\gamma_{v_0})\in L^{1}(\mathbb{R}^3)$, then for any $\varepsilon > 0$,
\begin{align}
\label{eq.field_decay_re} \| \Re ( \beta_t - \gamma_{v_\infty}^{ \bar{X}_0 + v_\infty t } ) \|_{ L^\infty } &\lesssim (1 + t)^{-\frac{3}{2} }\text{,} \\
\label{eq.field_decay_im} \| \Im ( \beta_t - \gamma_{v_\infty}^{ \bar{X}_0 + v_\infty t } ) \|_{ L^\infty } &\lesssim_\varepsilon (1 + t)^{-1 + \varepsilon} \text{.}
\end{align}
\end{itemize}
\end{thm}

Theorem \ref{thm:technical} is proved in Section \ref{sec.stabpf}.

\begin{remark}
If one assumes less (but not too little) spatial decay for $\beta_0 - \gamma_{v_0}$ in Theorem \ref{thm:technical}, then one can still prove a correspondingly weaker version of Theorem \ref{thm:technical}, in which one obtains a slower rate of convergence from $v_t$ to $v_\infty$.
\end{remark}

\section{Proof of Theorem \ref{thm.traveling_wave_reg}} \label{sec.inertpf}

In this section, we present the proof of Theorem \ref{thm.traveling_wave_reg}, i.e., we establish existence, uniqueness, and decay properties of solutions describing inertial particle motion.

\subsection{Setup and Outline} \label{sec.inertpf_intro}

Much of the strategy for constructing subsonic and sonic solutions is the same.
We begin the proof by substituting \eqref{eq:trawave} into \eqref{eq.emodel}, which yields the traveling wave equation
\begin{align}\label{eq:realtrawave}
- \i v \cdot \nabla \gamma_v = - \lapl \gamma_v + \real \gamma_v + W \text{.}
\end{align}
Making the change of variables $G_v = U_r^{-1} \gamma_v$ (see Section \ref{sec.equations_cov}) and recalling the notation \eqref{eq.Hv}, Eq. \eqref{eq:realtrawave} becomes the complex-linear equation
\begin{equation} \label{eq.traveling_wave_alt} H_v G_v = - W \text{.} \end{equation}
As the transformation from $\gamma_v$ to $G_v$ is invertible, it suffices to solve this equation for $G_v$.

Now, one must invert $H_v$ in order to solve \eqref{eq.traveling_wave_alt} for $G_v$, and hence $\gamma_v$, in terms of $W$.
After Fourier transformation, \eqref{eq.traveling_wave_alt} becomes
\begin{equation} \label{eq.traveling_wave_falt} h_v (\xi) \hat{G}_v (\xi) = - \hat{W} (\xi) \text{,} \end{equation}
where $h_v (\xi)$ is the symbol of $H_v$, given in \eqref{eq.hv}.

Since $h_v$ has a (sharp) lower bound $| \xi |^{-1} | h_v (\xi) | \geq 1 - | v |$, one expects different phenomena in the cases $|v| < 1$, $|v| = 1$, and $|v| > 1$.
To be more specific, consider the operator $\tilde{H}_v = ( - \lapl )^{-1/2} H_v$ which has symbol $| \xi |^{-1} h_v (\xi)$.
Observe that:
\begin{itemize}
\item $0$ lies in the resolvent set of $\tilde{H}_v$ when $|v| < 1$.

\item $0$ lies at the tip of the absolutely continuous spectrum of $\tilde{H}_v$ when $|v| = 1$.

\item $0$ is embedded in the continuous spectrum of $\tilde{H}_v$ when $|v| > 1$.
\end{itemize}
It is thus necessary to treat the subsonic, sonic, and supersonic cases separately when solving Eq. \eqref{eq.traveling_wave_alt} for $G_v$.

If the solution $G_v$ exists, then we can determine the regularity of $G_v$ and $U G_v$.
Since $\real \gamma_v = \real U G_v$ and $\imag \gamma_v = \imag G_v$, this will yield the desired regularities for $\gamma_v$ in the subsonic and sonic settings.

With $G_v$, and hence $\gamma_v$, in hand, we proceed to establish another important property of inertial solutions: namely that the velocity $v_t$ of the particle is constant,  $\equiv v$.
This is equivalent to showing that $\dot{P}_t$ vanishes identically.
To prove this, we turn to the second equation in \eqref{eq.emodel_alt}.
In particular, we show that the integral on the right-hand side necessarily vanishes when $B_t (x)$ is replaced by $G_v (x - vt)$.
This demonstrates that we indeed have inertial solutions.

The final step is to establish decay estimates for $G_v$ and $U G_v$, and hence for $\gamma_v$.
This analysis focuses primarily on the behavior of $\hat{G}_v (\xi)$ near the origin in $\xi$- space, as the strength of the singularity there determines the optimal rate of decay.
In particular, we expect less decay in the sonic case, as compared to the subsonic case, because the symbol $h_v^{-1}$ is more singular near the origin (at least along the $v$-direction) in the sonic case than in the subsonic case.
Furthermore, since the Fourier transform of $\Im \gamma_v$ is more singular near at the origin than the Fourier transform of $\Re \gamma_v$, we expect $\Re \gamma_v$ to have better decay than $\Im \gamma_v$.

Roughly speaking, every power $|x|^{-1}$ in spatial decay corresponds to one derivative in Fourier space, and every such derivative makes the behavior of the Fourier transform near the origin worse.
This observation results in intrinsic upper bounds on the fastest possible rates of decay of $\gamma_v$.

Since the best decay rates tend to involve fractional powers of $|x|^{-1}$, we must perform a somewhat subtle analysis of those decay rates.
The idea is to obtain optimal fractional powers by interpolating between adjacent integer powers.
For this, we derive, in Appendix \ref{sec.decay_spatial}, \emph{frequency-localized} decay estimates, which yield precise information for $G_v$ restricted to each individual Fourier band $|\xi| \simeq 2^k$.
When restricted to these bands, the aforementioned interpolation process is trivialized.
\footnote{Indeed, on the individual Fourier bands, these interpolation estimates reduce to H\"older's inequality; see Propositions \ref{thm.spatial_decay} and \ref{thm.spatial_decay_ex}.}
We can then piece together these estimates on the individual Fourier bands into the desired decay estimates.

\subsection{The Subsonic Case}

Here, we assume that $|v| < 1$.
From our assumptions on $W$, we know that $W \in H_a^b (\R^3)$, for every $a > -3/2$.
Since
\[ h_v ( \xi )^{-1} \lesssim ( 1 - | v | + | \xi | )^{-1} | \xi |^{-1} \lesssim \begin{cases} | \xi |^{-1} & | \xi | \leq 1 \text{,} \\ | \xi |^{-2} & | \xi | \geq 1 \text{,} \end{cases} \]
we can solve for $G_v \in H_{a + 1}^{b + 2} (\R^3)$, as $\hat{G}_v = - h_v^{-1} \hat{W}$, since
\begin{align*}
\| G_v \|_{ H_{a + 1}^{b + 2} }^2 &\lesssim \int_{ | \xi | \leq 1 } | \xi |^{2a + 2} \frac{ | \hat{W} (\xi) |^2 }{ | h_v (\xi) |^2 } d \xi + \int_{ | \xi | \geq 1 } | \xi |^{2b + 4} \frac{ | \hat{W} (\xi) |^2 }{ | h_v (\xi) |^2 } d \xi \lesssim \| W \|_{ H_a^b }^2 \text{.}
\end{align*}
A similar estimate shows that
\[ \| U G_v \|_{ H_a^{b + 2} } \lesssim \| G_v \|_{ H_{a + 1}^{b + 2} } \lesssim \| W \|_{ H_a^b } < \infty \text{.} \]
This implies the regularity properties \eqref{eq.traveling_wave_subsonic} for $\gamma_v$.

To show that $\dot{P}_t$ vanishes, we make use of \eqref{eq.emodel}; in particular,
\begin{align}
\label{eq.Pt_eq} \dot{P}_t = \langle \nabla W^{vt}, \real \gamma_v^{vt} \rangle = \real \left[ i \int_{\R^3} \xi \frac{ |\xi| }{ \langle \xi \rangle } \frac{1}{ h_v (\xi) } | \hat{W} (\xi) |^2 d \xi \right] \text{.}
\end{align}
Since $W \in L^1 (\R^3) \cap H^b (\R^3)$ and $b \geq -1/2$, the integral on the right-hand side is finite.
Moreover, the quantity on the right-hand side within the bracket is purely imaginary, hence $\dot{P}_t$ vanishes identically.

It remains to exhibit the expected decay of $\gamma_v$.
Since $G_v = - H_v^{-1} W$, we can apply Proposition \ref{thm.spatial_decay}, with $f = W$ and $M = H_v^{-1}$.
Since the symbol $h_v^{-1}$ of $H_v^{-1}$ is in the class $\mc{M}_{-1}^{-2}$, Eq. \eqref{eq.spatial_decay} implies that
\[ | \langle x \rangle^\mu G_v (x) | \lesssim \| \langle x \rangle^4 W \|_{ L^2 } \text{,} \qquad \mu \in [0, 2) \text{.} \]
The above implies the same decay rate for $\Im \gamma_v = \Im G_v$.
On the other hand, since $U G_v = - U H_v^{-1} W$, and symbol $U H_v^{-1}$ is now in the class $\mc{M}_0^{-2}$, we have that
\[ | \langle x \rangle^\mu U G_v (x) | \lesssim \| \langle x \rangle^5 W \|_{ L^2 } \text{,} \qquad \mu \in [0, 3) \text{.} \]
This implies the expected decay rate for $\Re \gamma_v = \Re U G_v$ and proves \eqref{eq.traveling_wave_subsonic_decay}.

\subsection{The Sonic Case}

To understand inertial motion at the speed of sound, i.e., for $|v| = 1$, is trickier than for $|v| < 1$, because the quantity $| \xi |^{-1} h_v (\xi)$ is now singular at the origin.
However, this singularity only occurs along one ``bad" direction, the $v$-direction.
Thus, by leveraging the remaining ``good" directions, we still obtain regular inertial solutions.

We again solve for $G_v$, as $\hat{G}_v = - h_v^{-1} \hat{W}$, by requiring that the right-hand side belongs to an appropriate function space.
As $| h_v ( \xi ) |^{-1} \lesssim | \xi |^{-2}$ in the region $| \xi | \geq 1$, we have that
\[ \int_{ | \xi | \geq 1 } | \xi |^{ 2b + 4} \frac{ | \hat{W} (\xi) |^2 }{ | h_v (\xi ) |^2 } d \xi \lesssim \int_{ | \xi | \geq 1 } | \xi |^{2b} | \hat{W} (\xi) |^2 d \xi < \infty \text{.} \]
For the low-frequency region, we have that
\begin{align*}
\int_{ | \xi | \leq 1 } | \xi |^{2 a + 2 } \frac{ | \hat{W} (\xi) |^2 }{ | h_v (\xi) |^2 } d \xi &\leq \| \hat{W} \|_{ L^\infty }^2 \int_{ | \xi | \leq 1 } \frac{ | \xi |^{2 a + 2} }{ | h_v (\xi) |^2 } d \xi \\
&\lesssim \| W \|_{ L^1 }^2 \int_{ | \xi | \leq 1 } \frac{ | \xi |^{2 a + 2} }{ | h_v (\xi) |^2 } d \xi ,
\end{align*}\\
with $a \in (-1/2, 1/2)$,

It remains to show that the integral on the right-hand side is finite.
For this purpose, we introduce spherical coordinates, $(\rho, \theta, \phi)$, on the unit ball, with $\phi \in [0, \pi]$ measuring the angle from the $v$-direction.
Then,
\begin{align*}
\int_{ |\xi| \leq 1 } \frac{ | \xi |^{2 a + 2} }{ | h_v (\xi) |^2 } \d \xi &= \int_0^{2 \pi} \int_0^\pi \int_0^1 \frac{ \rho^{ 2 a + 2 } \sin \phi }{ | \langle \rho \rangle - \cos \phi |^2 } \d \rho \d \phi \d \theta \\
&\lesssim \int_0^\pi \int_0^1 \frac{ \rho^{ 2 a + 2 } \phi }{ ( \rho^2 + \phi^2 )^2 } \d \rho \d \phi \text{,}
\end{align*}
where, in the last step, we used that $\sin \phi \lesssim \phi$, $1 - \cos \phi \simeq \phi^2$, and $\langle \rho \rangle - 1 \simeq \rho^2$ on our domain.
By reparametrizing $\rho = r \cos \alpha$ and $\phi = r \sin \alpha$ (or by integrating directly), we see that the above integral is finite whenever $a > -1/2$.

The above reasoning implies that $G_v \in H_{a+1}^{b+2} (\R^3)$, and hence $U G_v \in H_a^{b+2} (\R^3)$, for any $a > -1/2$.
From this, it follows that $\real \gamma_v$ and $\imag \gamma_v$ satisfy \eqref{eq.traveling_wave_sonic}.

To show that $\dot{P}_t$ vanishes, we again take a look at the right-hand side of \eqref{eq.Pt_eq}.
Like in the subsonic case, it suffices to show that the integral in \eqref{eq.Pt_eq} is finite.
The high-frequency region $|\xi| \geq 1$ is treated in the same fashion as in the subsonic case.
For the low frequencies, $|\xi| \leq 1$, we change coordinates, as before:
\begin{align*}
\left| \int_{ |\xi| \leq 1 } \xi \frac{ |\xi| }{ \langle \xi \rangle } \frac{1}{ h_v (\xi) } | \hat{W} (\xi) |^2 d \xi \right| &\lesssim \| W \|_{ L^1 }^2 \int_0^{2 \pi} \int_0^\pi \int_0^1 \frac{ \rho^3 \sin \phi }{ | \langle \rho \rangle - \cos \phi | } \d \rho \d \phi \d \theta \\
&\lesssim \| W \|_{ L^1 }^2 \int_0^\pi \int_0^1 \frac{ \rho^3 \phi }{ ( \rho^2 + \phi^2 ) } \d \rho \d \phi \\
&< \infty \text{.} 
\end{align*}

To exhibit the expected decay rates, we let $S_v$ be defined as in \eqref{eq.Sv}, i.e., as the Fourier multiplier operator with symbol $( \langle \xi \rangle - v \cdot \hat{\xi} )^{-1}$.
Writing
\[ G_v = - (- \lapl)^{-1/2} S_v W \text{,} \]
we may apply Proposition \ref{thm.spatial_decay_ex}, with $f = W$ and $M = (-\lapl)^{-1/2}$.
Since the symbol $|\xi|^{-1}$ of $M$ is in the class $\mc{M}_{-1}^{-1}$, Eq. \eqref{eq.spatial_decay_ex} implies that
\[ | \langle x \rangle^\mu G_v (x) | \lesssim \| \langle x \rangle^3 W \|_{ L^2 } \text{,} \qquad \mu \in [0, 2/3) \text{.} \]
Moreover, since $U G_v = - (1 - \lapl)^{-1/2} S_v W$, and the symbol $\langle \xi \rangle$ of $( 1 - \lapl )^{-1/2}$ is in the class $\mc{M}_0^{-1}$, then \eqref{eq.spatial_decay_ex} yields
\[ | \langle x \rangle^\mu U G_v (x) | \lesssim \| \langle x \rangle^3 W \|_{ L^2 } \text{,} \qquad \mu \in [0, 1) \text{.} \]
The above decay rates for $G_v$ and $U G_v$ imply \eqref{eq.traveling_wave_sonic_decay}.

\subsection{The Supersonic Case}

Finally, if $|v| > 1$, it follows from the definition of the energy \eqref{eq.energy} that it suffices to show that $\Re \gamma_v$ does not belong to $L^2 (\R^3)$.
This is easily seen from the following expression:
\[ \widehat{ \Re \gamma_v } (\xi) = \widehat{ \Re U G_v } (\xi) = \frac{ \hat{W} (\xi) }{ 1 + | \xi |^2 - ( v \cdot \hat{\xi} )^2 } \text{.} \]
Since the denominator $1 + | \xi |^2 - ( v \cdot \hat{\xi} )^2$ vanishes on a two-dimensional surface, $\real \gamma_v$ does not have finite energy under our assumptions on $W$.

\section{Proof of Theorem \ref{thm:technical}} \label{sec.stabpf}

In this section, we present the proof of Theorem \ref{thm:technical}, i.e., we establish asymptotic stability of subsonic inertial particle motion under suitable hypotheses on the initial conditions.

\subsection{Setup and Outline of Proof} \label{sec.stabpf_intro}

By Theorem \ref{thm.traveling_wave_reg}, there exist inertial solutions $\gamma_v$, with $v \in \R^3$ and $|v| < 1$.
Next, we define
\[ G_v = U_r^{-1} \gamma_v = - H_v^{-1} W  \]
as in Sections \ref{sec.equations_cov} and \ref{sec.inertpf_intro}.

By Theorem \ref{thm.emodel_gwp}, there is a global solution $(X, P, \beta)$ to Eq.\ \eqref{eq.emodel}, with
\[ X_t, P_t \in C ( [0, \infty); \R^3 ) \text{,} \qquad (\real \beta_t, \imag \beta_t) \in C ( [0, \infty); H^1 (\R^3) \times \dot{H}^1 (\R^3) ) \text{.} \]
As in Section \ref{sec.equations_cov}, we define
\[ B_t = U_r^{-1} \beta_t \text{,} \qquad B \in C ( [0, \infty); \dot{H}^1 (\R^3) ) \text{.} \]

Recall that we assumed $|v_0| < 1$.
Suppose $T > 0$ is such that $|v_t| < 1$ for all $t \in [0, T)$, that is, the motion of the particle remains subsonic in this time interval.
It is then reasonable to introduce a field $\delta_{t}$ by setting
\begin{equation} \label{eq.delta} \beta_t = \gamma_{v_t}^{X_t} + \delta_t \text{,} \qquad t \in [0, T) \text{.} \end{equation}
Note that $\delta_t$ describes the deviation of the field $\beta_{t}$ from the inertial solution $\gamma_{v_t}^{X_{t}}$.
As for $\beta_t$, we transform $\delta_t$ by defining $D_t = U_r^{-1} \delta_t$.

We now cast the equations of motion for the field $\beta_t$ in a more convenient form, in terms of $D_t$, as long as the particle motion remains subsonic.
Using the equations \eqref{eq.emodel_alt} for $B_t$ and \eqref{eq.traveling_wave_alt} for $G_{v_t}$, we derive the following evolution equation for $D$,
\begin{align}
\label{eq.emodel_ansatz} \i \dot{D}_t &= H D_t - \i ( \partial_v G_{v_t} )^{X_t} \cdot \dot{v}_t = H D_t + H_{v_t}^{-2} \nabla W^{X_t} \cdot \dot{v}_t \text{,} \qquad t \in [0, T) \text{.}
\end{align}
By Duhamel's principle, this equation implies the following integral equation:
\begin{equation} \label{eq.D_duhamel} D_t = e^{- \i H t } D_0 - \i \int_0^t e^{ -\i H (t - s) } H_{v_s}^{-2}\nabla W^{X_s}\cdot \dot{v}_s \d s \text{.} \end{equation}
The equation for the particle acceleration, $\dot{v}_t$, reads
\begin{align}
\label{eq:acceleration} \dot{v}_t = \int_{ \R^3 } \nabla W^{X_t} U \real D_t \text{.}  
\end{align}

At this point, the proof of Theorem \ref{thm:technical} consists of two main steps:
\begin{enumerate}
\item The first step is to show that the acceleration $\dot{v}_t$ tends to 0, as $t \rightarrow \infty$, and that the particle moves at a subsonic speed for all times.
This will imply, in particular, the existence of an asymptotic (subsonic) velocity $v_\infty$.

\item The second step involves showing that the field $\beta_t$ converges to the field of the corresponding inertial solution, $\gamma_{v_\infty}^{\bar{X}_0 + v_\infty t}$.
\end{enumerate}
These steps are carried out in Sections \ref{sec.acceleration} and \ref{sec.decay_field}, respectively.
Here, we briefly outline the basic ideas underlying our proof.

For the first step, we will derive an integral equation for $\dot{v}$.
Since this equation is based on the ansatz \eqref{eq.delta} and equations \eqref{eq.D_duhamel} and \eqref{eq:acceleration}, one has to impose a bootstrap assumption on $\dot{v}$ that is adapted to its expected decay rate.
In particular, this assumption justifies \eqref{eq.delta} up to some time $T>0$.
Starting from \eqref{eq.D_duhamel} and \eqref{eq:acceleration}, we derive an integral equation (with memory)
\begin{align}
\label{eq.vdot_rel_pre} \dot{v}_t = r (t) - \int_0^t M (t-s) \dot{v}_s \d s \text{,} \qquad t \in [0, T) \text{,}
\end{align}
where $M$ is a $3 \times 3$ matrix-valued function.
This equation has the general form $(1 + L) \dot{v} = r$, where $L$ is an integral operator.
The goal is to show that $1 + L$ can be inverted, and that $r$ is small in a suitable sense.

In order to invert $1 + L$, we perform a Fourier transform in the time variable.
Because $L \dot{v}$ is defined by a convolution in the time variable of $M$ with $\dot{v}$, it becomes a product after Fourier transformation to frequency space.
Formally, $\dot{v}$ can be expressed in terms of $r$ as follows:
\[ \hat{\dot{v}} = ( 1 + \hat{M} )^{-1} \hat{r} \text{.} \]
By inverting the Fourier transform, we succeed in solving for and controlling $\dot{v}$ in terms of $r$.
We note that all these steps need to be justified rigorously---their validity is not apparent, a priori.

In the proof, we have to accomplish the following tasks:
\begin{itemize}
\item First, we prove decay estimates for $M$.
Among other things, these estimates validate the Fourier transform of \eqref{eq.vdot_rel_pre}, as well as the subsequent inverse Fourier transform, which yields an expression of $\dot{v}$ in terms of $r$.

\item Next, one shows that $1 + \hat{M}$ is an invertible matrix.
Here, we use the spherical symmetry of $W$ to show that $\hat{M}$, with respect to an appropriate orthonormal frame, is diagonal.
Subsequently, we use some special properties of $M$ to show that either $1 + \hat{M}$ is positive-definite, or $\Im \hat{M} \neq 0$.

\item Finally, we derive decay estimates for $r$:
\[ r (t) \lesssim (1 + t)^{-4} \text{,} \qquad t \geq 0 \text{.} \]
\end{itemize}
Combining these ingredients, we will prove the following decay estimate for $\dot{v}$:
\[ \vert \dot{v}_t \vert \lesssim (1 + t)^{-4} \text{,} \qquad t \geq 0 \text{.} \]

Much of the analytical work revolves around proving decay estimates for $M$ and $r$.
Here we rely on the following observations:
\begin{itemize}
\item Each matrix element of $M$ and term of $r$ can be expressed in the form
\[ \int_{ \R^3 } e^{-i t h_v (\xi)} |\xi|^l f (\xi) d \xi \text{,} \]
where $v \in \R^3$ is subsonic, i.e., $|v| < 1$.

\item Since $|v| < 1$, the phase $h_v$ has no stationary points away from the origin.
In particular, the radial derivative of $h_v$ is uniformly bounded from below.
\end{itemize}
As a result, one can apply standard integration by parts arguments in the radial variable, with each integration by parts resulting in a factor of $t^{-1}$.
The number of integrations by parts one can take depends on the behavior of the integrand near the origin in frequency space. This is captured by the exponent $l$.
Consequently, the best decay rate one finds for any one of these terms is determined by $l$.
The details underlying these decay estimates can be found in Appendix \ref{sec.decay_radial}.

For some of these terms, one can exploit additional symmetries of the integrand to show that there are no singularities at the origin in frequency space.
For such terms, one can derive an arbitrarily fast polynomial decay in $t$, because the number of integrations by parts one can perform is arbitrarily large.
This observation will be important in controlling $M$.
For details, see Appendix \ref{sec.decay_radial}.

The analysis sketched above enables us to obtain a strictly better estimate on $\dot{v}$ than the one imposed in the bootstrap assumption.
A standard continuity argument yields the desired $t^{-4}$-decay for $\dot{v}$.
This decay implies the existence of an asymptotic velocity $v_\infty$, as well as the rate of convergence of $v_t$ to $v_\infty$ stated in Theorem \ref{thm:technical}.
Moreover, we readily obtain the asymptotic initial position $\bar{X}_0$ appearing in the associated linear trajectory $t \mapsto \bar{X}_0 + v_\infty t$.

It remains to carry out the second main step, which consists of proving convergence of the field.
For this purpose, we study the field $\delta_t$ using the Duhamel formula \eqref{eq.D_duhamel}, and we attempt to control each term on the right-hand side in the space $L^\infty$.
The main tools used to carry out this step are dispersive estimates for the propagator, $e^{-itH}$ (see Section \ref{sec.equations_cov}), which are discussed in Appendix \ref{sec.decay_unif}.

The basic dispersive estimate comes from \cite{gustafson06}, where it is shown, roughly speaking, that if a function $f \in L^1 (\R^3)$ is localized in a single Fourier band, then $\| e^{-i t H} f \|_{ L^\infty }$ decays like $t^{-3/2}$.
From this basic estimate, one can derive similar decay estimates for $e^{-i t H} M f$, where $M$ is a differential operator whose symbol belongs to an appropriate class.
Just like in the integration by parts estimates of Appendix \ref{sec.decay_radial}, the amount of decay in $t$ depends on how singular the symbol of $M$ becomes near the origin in Fourier space.
This can be analyzed for rather general $M$ using interpolation arguments that are trivialized when one considers a single Fourier band.
We take full advantage of this flexibility when applying such estimates to several different operators $M$, as determined by \eqref{eq.D_duhamel} and the definitions of $D$ and $\delta$.

By piecing together decay estimates localized in Fourier space, we arrive at the expected decay estimates, which imply that $\delta_t = \beta_t - \gamma_{v_t}^{X_t}$ decays as desired.
Similar to Theorem \ref{thm.traveling_wave_reg}, since the Fourier transform of $\Re \delta_t$ is better behaved at the origin than the Fourier transform of $\Im \delta_t$, we expect better decay for $\Re \delta_t$ than for $\Im \delta_t$.

In order to complete the proof of Theorem \ref{thm:technical}, we must establish decay for $\beta_t - \gamma_{v_\infty}^{\bar{X}_0 + v_\infty t}$, rather than just for $\delta_t$.
To accomplish this, we combine the estimates alluded to above with a comparison between the traveling waves $\gamma_{v_t}$ and $\gamma_{v_\infty}$.
This last step is relatively straightforward, involving variants of Sobolev estimates for the $G_v$'s in terms of $W$.

\subsection{Behavior of the Acceleration at Large Times} \label{sec.acceleration}

We wish to prove that the acceleration of the particle decays in time, provided the initial deviation, $\delta_0$, of the field from a traveling wave is chosen small enough, as assumed in Theorem \ref{thm:technical}.
More precisely, we prove that
\begin{align}
\label{eq:vdecay} | \dot{v}_t | \lesssim \varepsilon_0 (1 + t)^{-4} \text{,} \qquad t \in [0, \infty) \text{,}
\end{align}
whence the bounds \eqref{eq:vtconv} and \eqref{eq:xtconv}.
We recall that the constant $\varepsilon_0$ appears in the hypotheses of Theorem \ref{thm:technical} and must be sufficiently small.

In order to prove \eqref{eq:vdecay}, we resort to a bootstrap argument, which proceeds as follows.
First, we fix a sufficiently small $\epsilon \propto \varepsilon_0$, depending on $|v_0|$ and $W$.
For an arbitrary but fixed $T > 0$, we make the bootstrap assumption
\begin{equation} \label{eq.bootstrap} \max_{0 \leq t < T} (1 + t)^4 | \dot{v}_t | \leq \epsilon \text{.} \end{equation}
From \eqref{eq:acceleration} we see that, since $\| \delta_0 \|_{ L^2 }$ is assumed to be very small (with respect to $\| W \|_{ \dot{H}^1 }$), \eqref{eq.bootstrap} automatically holds for small $T$.

The goal is to show that \emph{if \eqref{eq.bootstrap} holds, and if $\varepsilon_0$ (and hence $\epsilon$) is sufficiently small, then we can prove a strictly better estimate}
\begin{equation} \label{eq.bootstrap_ex} \max_{0 \leq t < T} (1 + t)^{-4} | \dot{v}_t | \leq \frac{1}{2} \epsilon \text{.} \end{equation}
A standard continuity argument then shows that \eqref{eq.bootstrap} (and hence \eqref{eq.bootstrap_ex}) holds\textit{ for all} $T > 0$.
This yields \eqref{eq:vdecay}.

The next step is to derive a relation for $\dot{v}_t$ on the interval $t \in [0, T)$.
Since $\epsilon$ is small, the bootstrap assumption \eqref{eq.bootstrap} implies that, for any $0 \leq t < T$,
\begin{equation} \label{eq:subsonic} | v_t | \lesssim | v_0 | + \int_0^t | \dot{v}_\tau | d \tau < v_\ast < 1 \text{,} \end{equation}
with $v_\ast$ independent of $T$.
\footnote{In other words, the particle remains ``uniformly subsonic".  In particular, in various estimates that depend on the velocity $v_{*}$, constants will not blow up.}
Consequently, the ansatz \eqref{eq.delta} (and hence \eqref{eq.emodel_ansatz}-\eqref{eq:acceleration}) remains appropriate for all $t \in [0, T)$.
Plugging \eqref{eq.D_duhamel} into \eqref{eq:acceleration} and using the identity $\nabla W^{X_t}= \e^{ -X_t \cdot \nabla } \nabla W$, we find
\begin{align*}
\dot{v}_t &=  r_0(t) - \Re \left\langle \nabla W, \int_0^t \i U \e^{-\i H (t-s) + (X_t - X_s) \cdot \nabla} H_{v_s}^{-2} \nabla W \cdot \dot{v}_s \d s \right\rangle \text{,} \\
r_0 (t) &= \Re \scalar{ \nabla W^{X_t} }{ U \e^{-\i Ht} D_0 } \text{.}
\end{align*}
Next, we expand the second term in the equation for $\dot{v}$ around $v_0$.
Recalling the notation \eqref{eq.Hv}, we can expand the exponential factor as
\begin{align*}
\e^{ -\i H (t-s) + (X_t - X_s) \cdot \nabla } &= \e^{-\i H (t-s) + \int_s^t (v_\tau - v_0) \d \tau \cdot \nabla + (t-s) v_0 \cdot \nabla} \\
&= \e^{- \i (t-s) H_{v_0} } + \e^{- \i (t-s) H_{v_0} } ( \e^{ \int_s^t (v_\tau - v_0) \d \tau \cdot \nabla } - 1 ) \text{.}
\end{align*}
From the above considerations we arrive at
\begin{align}
\label{eq:vdotM} \dot{v}_t &= r (t) - \int_0^t M (t-s) \dot{v}_s \d s \text{,} \qquad t \in [0, T) \text{,}
\end{align}
where $r (t) = r_0 (t) - r_1 (t) - r_2 (t)$, where
\begin{align*}
r_1 (t) &= \Re \left\langle \nabla W, \int_0^t \i U \e^{- \i (t-s) H_{v_0} } ( \e^{\int_s^t (v_\tau - v_0) \d \tau \cdot \nabla} - 1 ) H_{v_s}^{-2} \nabla W \cdot \dot{v}_s \d s \right\rangle \text{,} \\
r_2 (t) &= \Re \left\langle \nabla W, \int_0^t \i U \e^{- \i (t-s) H_{v_0} } ( H_{v_s}^{-2} - H_{v_0}^{-2} ) \nabla W \cdot \dot{v}_s \d s \right\rangle \text{,}
\end{align*}
and where $M$ is the matrix-valued function
\begin{align*}
M_{ij} (t) = \begin{cases} \Re \scalar{ \partial_{x_i} W }{ \i U \e^{-\i t H_{v_0} } H_{v_0}^{-2} \partial_{x_j} W} & t \geq 0 \text{,} \\ 0 & t < 0 \text{.} \end{cases}
\end{align*}

Equation \eqref{eq:vdotM} can be expressed in the form
\[ (1 + L) \dot{v} = r \text{,} \]
where $L$ is the linear operator
\[ Lh (t) = \int_0^t M (t-s) h (s) d s \text{.} \]
The key observation, stated in the following lemma, is that $1 + L$ is invertible.

\begin{lemma} \label{lem:invL}
There exists a $3 \times 3$ matrix function $K$ on $\R$ satisfying
\begin{itemize}
\item $K (t)$ vanishes for all $t < 0$.

\item The following decay estimate holds:
\begin{equation} \label{eq:Kdecay} | K (t) | \lesssim_W (1 + t)^{-4} \text{.} \end{equation}

\item The following identity holds:
\begin{align}
\label{eq:kernelForm} \dot{v}_t = r (t) + \int_0^t K (t-s) r (s) ds \text{.}
\end{align}
\end{itemize}
\end{lemma}

\begin{proof}
See Section \ref{sec:invL}.
\end{proof}

In addition, we will also need the following estimates on $r$.

\begin{lemma} \label{lem:hdecay}
For any $t \in [0, T)$, the functions $r_i (t)$ satisfy
\begin{align}
\label{eq:hdecay_r0} | r_0 (t) | &\lesssim_W (1 + t)^{-4} \norm{ \langle x \rangle^4 \delta_0 }_{ L^2 } \leq \varepsilon_{0} (1+t)^{-4},\\
\label{eq:hdecay_r1} | r_1 (t) | &\lesssim_W \epsilon^2 (1 + t)^{-4} \text{,} \\
\label{eq:hdecay_r2} | r_2 (t) | &\lesssim_W \epsilon^2 (1 + t)^{-4} \text{.}
\end{align}
\end{lemma}

\begin{proof}
See Section \ref{sec:remainder}.
\end{proof}

Applying Lemma \ref{lem:hdecay} and recalling the assumptions on $\delta_0$ and $W$ in Theorem  \ref{thm:technical}, we obtain
\[ | r (t) | \lesssim \varepsilon_0 ( 1 + t )^{-4} + \epsilon^2 ( 1 + t )^{-4} \text{.} \]
Using \eqref{eq:Kdecay} and \eqref{eq:kernelForm}, we have that
\begin{align*}
| \dot{v}_t | &\lesssim ( \varepsilon_0 + \epsilon^2 ) \int_0^t (1 + t - s)^{-4} (1 + s)^{-4} ds + ( \varepsilon_0 + \epsilon^2 ) (1 + t)^{-4} \\
&\lesssim ( \varepsilon_0 + \epsilon^2 ) (1 + t)^{-4} \text{,}
\end{align*}
for all $t<T$.
Choosing $\epsilon \propto \varepsilon_0$ sufficiently small, we arrive at the strictly better estimate
\[ | \dot{v}_t | \leq \frac{\epsilon}{2} (1 + t)^{-4} \text{,} \]
provided $\varepsilon_0$ is chosen sufficiently small.
This completes the bootstrap argument and hence proves \eqref{eq:vdecay}.

The bounds \eqref{eq:vtconv} and \eqref{eq:xtconv} follow easily from \eqref{eq:vdecay}.
By \eqref{eq:vdecay}, $| \dot{v} |$ is integrable on $[0, \infty)$. Hence there exists an asymptotic velocity,
\[ v_\infty = v_0 + \lim_{t \nearrow \infty} \int_0^t \dot{v}_{s} ds \text{.} \]
Our hypotheses on $v_0$ and our bounds on $\dot{v}$ imply that $| v_\infty | < 1$
and show that
\begin{align*}
| v_t - v_\infty | \leq \int_t^\infty | \dot{v}_s | ds \lesssim \varepsilon_0 (1 + t)^{-3} \text{,}
\end{align*}
which proves \eqref{eq:vtconv}.

To prove existence of an asymptotic initial position, $\bar{X}_0$, we show that
\[ Y_t = X_t - v_{\infty} t \]
has a limit, as $t \nearrow \infty$.
For any $0 < t_1 < t_2$,
\begin{align*}
Y_{t_2} - Y_{t_1} = \int_{t_1}^{t_2} ( v_s - v_\infty ) ds \text{.}
\end{align*}
This, together with \eqref{eq:vtconv}, implies that $| Y_{t_2} - Y_{t_1} | \rightarrow 0$, as $t_1, t_2 \nearrow \infty$.
It follows that $Y_t$ has a limit $\bar{X}_0$, as $t \nearrow \infty$, and
\begin{align}\label{eq:barX}
| Y_t - \bar{X}_0 | \leq \int_{t}^{\infty} | v_s - v_\infty | ds \lesssim \varepsilon_0 (1 + t)^{-2} \text{,}
\end{align}
which completes the proof of \eqref{eq:xtconv}.

\subsubsection{Proof of Lemma \ref{lem:invL}} \label{sec:invL}

The main idea is to apply the Fourier transform in the $t$-variable to \eqref{eq:vdotM}, but some care is required to justify this.
We begin by exhibiting some rather straightforward properties of $M$.
We note that
\begin{align}
\label{eq:Momega} M_{ij} (t) &= \Re \scalar{ \xi_i \hat{W} (\xi) }{ \i | \xi | \langle \xi \rangle^{-1} \e^{-\i t h_{v_0} (\xi)} h_{v_0}^{-2} (\xi) \xi_j \hat{W} (\xi) } \text{,} \\
\notag \hat{M}_{ij} (\omega) &= \int_0^\infty \e^{- \i \omega t} \Re \scalar{ \xi_i \hat{W} (\xi) }{ \i | \xi | \langle \xi \rangle^{-1} \e^{-\i t h_{v_0} (\xi)} h_{v_0}^{-2} (\xi) \xi_j \hat{W} (\xi) } \d t
\end{align}

\begin{lemma} \label{lem:Momega}
The following properties hold.
\begin{itemize}
\item $M$ satisfies the estimate
\begin{equation} \label{eq.M_est} | M (t) | \lesssim_W ( 1 + t )^{-6} \text{,} \qquad t \in [0, \infty) \text{.} \end{equation}

\item $\hat{M}$ extends to an analytic function in the lower half-plane, $\Pi_-$.
Moreover, this extension is continuous on the closure $\bar{\Pi}_-$.

\item For any $0 \leq l \leq 4$ and any $z = \omega - i y \in \bar{\Pi}_- \setminus \{ 0 \}$,
\begin{equation} \label{eq.Mhat_est} \left| \left( \frac{d}{d \omega} \right)^l \hat{M} (\omega - i y) \right| \lesssim_W ( 1 + | z | )^{-2} \text{.} \end{equation}
\end{itemize}
\end{lemma}

\begin{proof}
The decay estimate \eqref{eq.M_est} follows from repeated integrations by parts. More precisely, we apply Proposition \ref{lem:decayestimates}, with $f = W$, $v = v_0$, and $N = 6$.
In particular, $M$ is supported on the half-line $[0, \infty)$ and belongs to $L^2(\mathbb{R}) \cap L^1(\mathbb{R})$.
This implies that $\hat{M}$ can be analytically continued to the lower-half plane $\Pi_-$, and that it is continuous on $\bar{\Pi}_-$.

Next, for any $z = \omega - i y \in \CC \setminus \{ 0 \}$, with $y \leq 0$, Eq. \eqref{eq:Momega} implies that
\begin{align*}
z \hat{M}_{ij} (z) &= \i \int_0^\infty [ \partial_t \e^{-\i z t} ] \Re \scalar{ \xi_i \hat{W} (\xi) }{ \i | \xi | \langle \xi \rangle^{-1} \e^{-\i t h_{v_0} (\xi)} h_{v_0}^{-2} (\xi) \xi_j \hat{W} (\xi) } \d t \\
&= -\i \Re \scalar{ \xi_i \hat{W} (\xi) }{ \i | \xi | \langle \xi \rangle^{-1} h_{v_0}^{-2} (\xi) \xi_j \hat{W} (\xi) } \\
&\qquad - \i \int_0^\infty \e^{-\i z t} \Re \scalar{ \xi_i \hat{W} (\xi) }{ | \xi | \langle \xi \rangle^{-1} \e^{-\i t h_{v_0} (\xi)} h_{v_0}^{-1} (\xi) \xi_j \hat{W} (\xi) } \d t \text{.}
\end{align*}
Using integration by parts (see Proposition \ref{lem:decayestimates}), one shows that the ($t$-)integrand in the second term on the right-hand side decays nicely, as $t \nearrow \infty$.
The first term on the right-hand side vanishes, since the $L^2$-inner product is purely imaginary.
Thus, we can integrate by parts again, and this yields
\begin{align*}
z^2 \hat{M}_{ij} (z) &= - \Re \scalar{ \xi_i \hat{W} (\xi) }{ | \xi | \langle \xi \rangle^{-1} h_{v_0}^{-1} (\xi) \xi_j \hat{W} (\xi) } \\
&\qquad + \int_0^\infty \e^{-\i z t} \Re \i \scalar{ \xi_i \hat{W} (\xi) }{ | \xi | \langle \xi \rangle^{-1} \e^{-\i t h_{v_0} (\xi)} \xi_j \hat{W} (\xi) } \d t \text{.}
\end{align*}
The second term on the right is well defined by Proposition \ref{lem:decayestimates}.
Moreover, by Plancherel's theorem, the first term on the right side is bounded by $\| W \|_{ L^2 }^2$.

This proves \eqref{eq.Mhat_est} for $l = 0$.
If $l > 0$, then
\[ | z^2 \partial_\omega^l \hat{M}_{ij} (z) | = \left| \int_0^\infty [ \partial_t^2 \e^{-\i z t} ] t^l \Re \scalar{ \xi_i \hat{W} (\xi) }{ \i | \xi | \langle \xi \rangle^{-1} \e^{-\i t h_{v_0} (\xi)} h_{v_0}^{-2} (\xi) \xi_j \hat{W} (\xi) } \d t \right| \text{.} \]
By \eqref{eq.M_est}, the integrand is integrable in $t$.
From here, we can integrate by parts twice, as in the $l = 0$ case, and complete the proof of \eqref{eq.Mhat_est}.
\end{proof}

Next, we define the auxiliary functions
\begin{align*}
E_v (t) = \begin{cases} \dot{v}_t & 0 \leq t < T \text{,} \\ 0 & t < 0 \text{ or } t \geq T \text{,} \end{cases} \qquad E_r (t) = \begin{cases} r (t) & 0 \leq t < T \text{,} \\ 0 & t < 0 \text{,} \\ \int_0^T M (t - s) \dot{v}_s ds & t \geq T \text{.} \end{cases}
\end{align*}
From \eqref{eq:vdotM}, we obtain the equation
\[ E_v = E_r - M \ast E_v \text{.} \]
By definition, $E_v$ is compactly supported.
Moreover, Lemma \ref{lem:Momega} implies that $E_r (t)$ is $O( t^{-6} )$, as $t \nearrow \infty$.
Thus, we can apply the Fourier transform to obtain
\begin{equation} \label{eq.fourier_pre} [ 1 + \hat{M} (\omega) ] \hat{E}_v (\omega) = \hat{E}_r (\omega) \text{,} \qquad \omega \in \R \text{.} \end{equation}

The matrix on the left-hand side turns out to be invertible.

\begin{lemma} \label{lem:matrixinvertible}
Let $z \in \C$, with $\Im z \leq 0$.
\begin{itemize}
\item If $\Re z > 0$, then $\Im \hat{M} (z)$ is strictly negative definite.

\item If $\Re z < 0$, then $\Im \hat{M} (z)$ is strictly positive definite.

\item If $\Re z = 0$, then $\hat{M} (z)$ is strictly positive definite.
\end{itemize}
In particular, $1 + \hat{M}$ is invertible everywhere on $\bar{\Pi}_-$.
\end{lemma}

\begin{proof}
Without loss of generality, we may assume that $v_0 = (0, 0, |v_0|)$.
From the explicit formula for $M$ in \eqref{eq:Momega} and from the assumption that $W$ is spherically symmetric, we conclude that $M$ is diagonal.
Thus, it suffices to show that each diagonal element $\hat{M}_{ii}$ has the desired sign.

We define the quantity
\[ Q_i (\xi) = \xi_i^2 | \xi | \langle \xi \rangle^{-1} | \hat{W} (\xi) |^2 | h_{v_0} (\xi) |^{-2} > 0 \text{,} \qquad \xi \neq 0 \text{.} \]
Note $Q_i$ is strictly positive, since $\hat{W}$ is assumed to be everywhere nonvanishing.
For any $z \in \CC$, with $\Im z < 0$, we have, by \eqref{eq:Momega}, that
\begin{align*}
\hat{M}_{ii} (z) &= \int_0^\infty \e^{-\i z t} \Re \int_{ \R^3 } \i Q_i (\xi) e^{-i t h_{v_0} (\xi) } \d \xi \d t \\
&= \frac{1}{2} \int_0^\infty \e^{-\i z t} \int_{ \R^3 } \i Q_i (\xi) [ e^{-i t h_{v_0} (\xi)} - e^{i t h_{v_0} (\xi)} ] \d \xi \d t \\
&= \frac{1}{2} \int_0^\infty \int_{ \R^3 } \i Q_i (\xi) \{ e^{-i t [ h_{v_0} (\xi) + z] } - e^{i t [ h_{v_0} (\xi) - z ] } \} \d \xi \d t \text{.}
\end{align*}
Because the integrand has good decay when $t \nearrow \infty$, we can integrate this directly:
\begin{align}
\label{eql.matrixinvertible_1} \hat{M}_{ii} (z) &= \frac{1}{2} \int_{ \R^3 } Q_i (\xi) \left[ \frac{1}{ h_{v_0} (\xi) + z } + \frac{1}{ h_{v_0} (\xi) - z } \right] d \xi \text{.}
\end{align}

Next, we write $z = \omega - i y$, with $y > 0$.
If $\omega = \Re z = 0$ then \eqref{eql.matrixinvertible_1} yields
\begin{align*}
\hat{M}_{ii} (z) &= \frac{1}{2} \int_{ \R^3 } Q_i (\xi) \left[ \frac{1}{ h_{v_0} (\xi) - y i } + \frac{1}{ h_{v_0} (\xi) + y i } \right] d \xi \\
&= \int_{ \R^3 } Q_i (\xi) \frac{ h_{v_0} (\xi) }{ h_{v_0}^2 (\xi) + y^2 } d \xi \\
&> 0 \text{.}
\end{align*}
From the above, we also see that $\hat{M}_{ii} (z)$ remains strictly positive even when $y \searrow 0$; hence the proof is complete in the case where $\Re z = 0$.

Next, if $\omega \neq 0$ and $y > 0$, then
\begin{align*}
\Im \hat{M}_{ii} (z) &= \frac{1}{2} \Im \int_{ \R^3 } Q_i (\xi) \left[ \frac{1}{ h_{v_0} (\xi) + \omega - y i } + \frac{1}{ h_{v_0} (\xi) - \omega + y i } \right] d \xi \\
&= \frac{1}{2} \Im \int_{ \R^3 } Q_i (\xi) y \left[ \frac{1}{ | h_{v_0} (\xi) + \omega |^2 + y^2 } - \frac{1}{ | h_{v_0} (\xi) - \omega |^2 + y^2 } \right] d \xi \text{.}
\end{align*}
Since $h_{v_0}$ is everywhere positive, it follows that $| h_{v_0} (\xi) + \omega |^2 > | h_{v_0} (\xi) - \omega |^2$ if $\omega > 0$, and $| h_{v_0} (\xi) + \omega |^2 < | h_{v_0} (\xi) - \omega |^2$ if $\omega < 0$.
It follows that $\Im \hat{M}_{ii} (z)$ is negative if $\omega > 0$, while $\Im \hat{M}_{ii} (z)$ is positive if $\omega < 0$.

Finally, we examine the limit $y \searrow 0$ when $\omega \neq 0$.
Writing
\begin{align*}
\Im \hat{M}_{ii} (\omega) &= \frac{1}{2} \lim_{y \searrow 0} \Im \int_{ \R^3 } Q_i (\xi) \left[ \frac{1}{ h_{v_0} (\xi) + \omega - y i } + \frac{1}{ h_{v_0} (\xi) - \omega + y i } \right] d \xi
\end{align*}
in polar coordinates and recalling the standard formula $\Im (a \mp i 0)^{-1} = \pm \pi^{-1} \delta (a)$ (in one dimension), we can write the integral above as
\begin{align*}
\Im M_{ii} (x) &= \frac{1}{2} \int_{ \R^3 } Q_i (\xi) \delta ( h_{v_0} (\xi) + \omega ) \d \xi - \frac{1}{2} \int_{ \R^3 } Q_i (\xi) \delta ( h_{v_0} (\xi) - \omega ) \d \xi \\
&=: A_+ - A_- \text{,}
\end{align*}
where $A_+$ is the integral of $Q_i$ over the set where $h_{v_0} = - \omega$ and $A_-$ is the integral of $Q_i$ over the set where $h_{v_0} = \omega$.
\footnote{Since $h_{v_0}$ does not have any critical points on $\R^3 \setminus \{ 0 \}$, these integrals are well-defined.}

Suppose first that $\omega > 0$.
Then, $A_+$ vanishes, because $h_{v_0}$ is everywhere positive on $\R^3 \setminus \{ 0 \}$.
In contrast, $A_-$ is strictly positive, since $Q_i > 0$, and since $h_{v_0}$ maps onto $(0, \infty)$.
Consequently, in the case $\omega > 0$, we conclude that
\[ \Im \hat{M}_{ii} (\omega) = A_+ - A_- < 0 \text{.} \]
On the other hand, if $\omega < 0$, similar reasoning shows that $A_+ > 0$, while $A_- = 0$; hence, $\Im \hat{M}_{ii} (\omega) > 0$.
\end{proof}

Applying Lemma \ref{lem:matrixinvertible} to Eq. (\ref{eq.fourier_pre}), we get
\begin{align}
\label{eq:halffourier} \hat{E}_v (\omega) &= [ 1 + \hat{M} (\omega) ]^{-1} \hat{E}_r (\omega) = \{ [ 1 + \hat{M} (\omega) ]^{-1} - 1 \} \hat{E}_r (\omega) + \hat{E}_r (\omega) \text{.}
\end{align}
Note that $( 1 + \hat{M} )^{-1}$ tends to $1$ at infinity, hence its inverse Fourier transform is singular. We therefore prefer to work with $( 1 + \hat{M} )^{-1} - 1$.

The matrix $1 + \hat{M}$ vanishes nowhere on the real line.
Moreover, since $\hat{M}$ tends to $0$, as $\omega \nearrow \infty$ (see \eqref{eq.Mhat_est}), it follows that $[ 1 + \hat{M} (\omega) ]^{-1}$ is uniformly bounded from above, for all $\omega \in \R$.
\footnote{The precise upper bound, though, depends on the properties of $W$.}
Thus, the inverse Fourier transform
\begin{align}
\label{eq:functionK} K (t) &= \int_{-\infty}^\infty \{ [ 1 + \hat{M} (\omega) ]^{-1} - 1 \} \e^{\i \omega t} \d \omega = - \int_{-\infty}^\infty \hat{M} (\omega) [ 1 + \hat{M} (\omega) ]^{-1} e^{i \omega t} d \omega
\end{align}
is well-defined.

To prove the bound \eqref{eq:Kdecay} in Lemma 5.1, we first estimate
\begin{align*}
\| \langle t \rangle^4 K \|_{ L^\infty } &\lesssim \sum_{k = 0}^4 \| \partial_\omega^l \{ \hat{M} (\omega) [ 1 + \hat{M} (\omega) ]^{-1} \} \|_{ L^1 } \lesssim \sum_{k = 0}^4 \| \partial_\omega^l \hat{M} (\omega) \|_{ L^1 } \text{,}
\end{align*}
where, in the last step, we have used the boundedness of $( 1 + \hat{M} )^{-1}$.
Applying H\"older's inequality and \eqref{eq.Mhat_est}, we find that
\begin{align*}
\| \langle t \rangle^4 K \|_{ L^\infty } &\lesssim \sum_{k = 0}^4 \| \langle \omega \rangle^2 \partial_\omega^l \hat{M} (\omega) \|_{ L^\infty } < \infty \text{.}
\end{align*}
This proves \eqref{eq:Kdecay}.

Since Lemma \ref{lem:matrixinvertible} implies that the integrand in \eqref{eq:functionK} is analytic on the lower half-plane, a contour integration argument shows that, for any $y > 0$,
\begin{align*}
K (t) = \e^{yt} \int_{-\infty}^\infty \left\{ \left[ 1 + \hat{M} (\omega - \i y) \right]^{-1} - 1 \right\} \e^{\i \omega t} \d \omega \text{.}
\end{align*}
If $t < 0$, then, letting $y \nearrow \infty$ and recalling \eqref{eq.Mhat_est}, one sees that $K (t) = 0$.

Finally, using \eqref{eq:Kdecay} and \eqref{eq:halffourier}, we see that
\[ E_v = K \ast E_r + E_r \text{.} \]
Since $K (t)$ vanishes when $t < 0$, this  yields \eqref{eq:kernelForm}.

\subsubsection{Proof of Lemma \ref{lem:hdecay}} \label{sec:remainder}

First, since $X_0 = 0$, we have that
\begin{equation} \label{eq:averageX} \frac{ | X_t | }{ t } \leq \frac{1}{t} \int_0^t |v_\tau| \d \tau < v_\ast < 1 \text{,} \end{equation}
where $v_\ast$ is as in \eqref{eq:subsonic}.
We use the Parseval identity to write $r_0 (t)$ as
\begin{align*}
r_0 (t) &= \Re \left\langle \i \xi \hat{W} (\xi), \e^{- \i t [ h (\xi) - \frac{X_t}{t} \cdot \xi ] } \widehat{U D}_0 (\xi) \right\rangle \\
&= \Re \left\langle \i \xi \hat{W} (\xi), \e^{- \i t h_{ \bar{v} } (\xi) } [ \widehat{\Re \delta}_0 (\xi) + | \xi | \langle \xi \rangle^{-1} \i \widehat{ \Im \delta_0 } ] \right\rangle \text{,}
\end{align*}
where $\bar{v} = t^{-1} X_t$.
If $t \leq 1$, we can apply H\"older's inequality to bound
\[ r_0 (t) \lesssim \| W \|_{ H^1 } \| \delta_0 \|_{ L^2 } \text{.} \]
Next, for $t > 1$, we apply Proposition \ref{thm.dispersive_subsonic}, with $v = \bar{v}$, and with $(l, b) = (1, 0)$ and $(l, b) = (2, -1)$, respectively.
Combining the two cases yields \eqref{eq:hdecay_r0}.

For $r_1$, we first expand
\begin{align*}
e^{ \int_s^t ( v_\tau - v_0 ) d \tau \cdot \nabla } - 1 &= \int_s^t \frac{d}{d \tau} e^{ \int_s^\tau ( v_{\tau^\prime} - v_0 ) d \tau^\prime \cdot \nabla } d \tau \\
&= \int_s^t e^{ \int_s^\tau ( v_{\tau^\prime} - v_0 ) d \tau^\prime \cdot \nabla } [ ( v_\tau - v_0 ) \cdot \nabla ] d \tau \\
&= \int_s^t e^{ \int_s^\tau ( v_{\tau^\prime} - v_0 ) d \tau^\prime \cdot \nabla } \left( \int_0^\tau \dot{v}_{ s^\prime } d s^\prime \cdot \nabla \right) d \tau \text{,}
\end{align*}
so that, by Parseval's identity,
\begin{align*}
r_1 (t) &= \real \int_0^t \int_s^t \int_0^\tau \i f_1 (t, s, \tau, s^\prime) d s^\prime d \tau d s \text{,} \\
f_1 (t, s, \tau, s^\prime) &= \langle \nabla W, U \e^{- \i (t-s) H_{v_0} } e^{ \int_s^\tau ( v_{\tau^\prime} - v_0 ) d \tau^\prime \cdot \nabla } \dot{v}_{s^\prime} \cdot \nabla H_{v_s}^{-2} \dot{v}_s \cdot \nabla W \rangle \\
&= \i \langle \xi \hat{W} (\xi), | \xi | \langle \xi \rangle^{-1} \e^{- \i (t-s) h_{ \bar{v} } (\xi) } ( \dot{v}_{s^\prime} \cdot \xi ) h_{v_s}^{-2} (\xi) ( \dot{v}_s \cdot \xi ) \hat{W} (\xi) \rangle \text{,} \\
\bar{v} &= v_0 + \frac{1}{t - s} \int_s^\tau ( v_{\tau^\prime} - v_0 ) d \tau^\prime \text{.}
\end{align*}
Note that $\bar{v}$ is subsonic, since
\[ | \bar{v} | \leq \frac{t - \tau}{t - s} | v_0 | + \frac{1}{t - s} \int_s^\tau | v_\tau^\prime | d \tau^\prime < v_\ast \text{.} \]

The idea is the same as that for $r_0$.
If $t \leq 1$, then we can trivially bound
\[ r_1 (t) \lesssim \| W \|_{ H^1 } \| W \|_{ L^2 } \text{.} \]
For $t > 1$, we note that the symbol
\[ g (\xi) = \langle \xi \rangle^{-1} \left( \frac{ \dot{v}_{s^\prime} }{ | \dot{v}_{s^\prime} | } \cdot \xi \right) h_{v_s}^{-2} (\xi) \left( \frac{ \dot{v}_s }{ | \dot{v}_s | } \cdot \xi \right) \]
belongs to the class $\mc{O}^{-3}$ (see Section \ref{sec.symbol}).
Consequently, applying Proposition \ref{thm.dispersive_subsonic}, with $l = 2$, $b = -3$, and $v = \bar{v}$, yields
\footnote{One can apply an argument similar to one in Proposition \ref{lem:decayestimates} to conclude even better decay for $r_1$.}
\begin{align*}
| f_1 (t, s, \tau, s^\prime) | &\lesssim ( 1 + t - s )^{-5} | \dot{v}_{s^\prime} | | \dot{v}_s | \| \langle x \rangle^5 W \|_{ L^2 }^2 \text{.}
\end{align*}

Similarly, for $r_2$, we can write
\begin{align*}
H_{v_s}^{-2} - H_{v_0}^{-2} &= \int_0^s \frac{d}{d \tau} H_{v_\tau}^{-2} d \tau = -2 i \int_0^s H_{v_\tau}^{-3} ( \dot{v}_\tau \cdot \nabla ) d \tau \text{,}
\end{align*}
which yields the identity
\begin{align*}
r_2(t) &= 2 \Re \int_0^t \int_0^s \scalar{ \nabla W }{ U \e^{- \i (t-s) H_{v_0} } H_{v_\tau}^{-3} ( \dot{v}_\tau \cdot \nabla ) ( \dot{v}_s \cdot \nabla ) W } \d \tau \d s \\
&= 2 \Re \i \int_0^t \int_0^s \scalar{ \xi \hat{W} (\xi) }{ | \xi | \langle \xi \rangle^{-1} \e^{ - \i (t-s) h_{v_0} (\xi) } h_{v_\tau}^{-3} (\xi) ( \xi \cdot \dot{v}_\tau ) ( \xi \cdot \dot{v}_s ) \hat{W} (\xi) } \d \tau \d s \text{.}
\end{align*}
Let $f_2 (t, s, \tau)$ denote the quantity within the above double integral.

When $t \leq 1$, we once again have
\[ r_2 (t) \lesssim \| W \|_{ L^2 }^2 \text{.} \]
In the case $t > 1$, since
\[ g (\xi) = | \xi | \langle \xi \rangle^{-1} \left( \frac{ \dot{v}_\tau }{ | \dot{v}_\tau | } \cdot \xi \right) h_{v_s}^{-3} (\xi) \left( \frac{ \dot{v}_s }{ | \dot{v}_s | } \cdot \xi \right) \]
is in the class $\mc{O}^{-4}$, we can apply Proposition \ref{thm.dispersive_subsonic}, with $l = 1$, $b = -4$, and $v = v_0$:
\[ | f_2 (t, s, \tau) | \lesssim (1 + t - s)^{-4} \| \langle x \rangle^4 W \|_{ L^2 }^2 \text{.} \]
By the bootstrap assumption for $\dot{v}$,
\begin{align*}
| r_2 (t) | &\lesssim \epsilon^2 \int_0^t \int_0^s (1 + t - s)^{-4} (1 + s)^{-4} (1 + \tau)^{-4} \d\tau \d s \lesssim \epsilon^2 (1 + t)^{-4} \text{.}
\end{align*}
This completes the proof of \eqref{eq:hdecay_r2}.

\subsection{Asymptotics of the Field, as $t$ tends to $\infty$} \label{sec.decay_field}

In this part of the proof, we examine the asymptotics of the field $\beta$.
In particular, we wish to prove the decay estimates \eqref{eq.field_decay_re} and \eqref{eq.field_decay_im}, and thereby complete the proof of Theorem \ref{thm:technical}.

\subsubsection{Preliminary Estimates}

We first collect several estimates that are needed in the proofs of \eqref{eq.field_decay_re} and \eqref{eq.field_decay_im}.
We begin with some $L^\infty$-estimates involving the operator $H_v^{-1}$ defined in \eqref{eq.Hv}.

\begin{lemma} \label{thm.trivial_unif_cor}
Let $v, v^\prime \in \R^3$, with $| v | < 1$ and $| v^\prime | < 1$, and let $f \in \mc{S} (\R^3)$.
\begin{itemize}
\item The following uniform bounds hold:
\begin{align}
\label{eq.trivial_unif_cor} \| U H_v^{-1} f \|_{ L^\infty } + \| H_v^{-1} f \|_{ L^\infty } &\lesssim_{ |v| } \| f \|_{ L^2 } \text{,} \\
\notag \| U H_v^{-1} H_{v^\prime}^{-1} \nabla f \|_{ L^\infty } + \| H_v^{-1} H_{v^\prime}^{-1} \nabla f \|_{ L^\infty } &\lesssim_{ |v|, |v^\prime| } \| f \|_{ L^2 } \text{.}
\end{align}

\item The following uniform bound holds:
\begin{equation} \label{eq.trivial_comp_cor} \| U H_v^{-1} f - U H_{v^\prime}^{-1} f \|_{ L^\infty } + \| H_v^{-1} f - H_{v^\prime}^{-1} f \|_{ L^\infty } \lesssim_{ |v|, |v^\prime| } | v - v^\prime | \| f \|_{ L^2 } \text{.} \end{equation}
\end{itemize}
\end{lemma}

\begin{proof}
We first note that the symbol, $m (\xi)$, of $U H_v^{-1}$ satisfies
\[ | m (\xi) | = \frac{ | \xi | }{\langle \xi \rangle |  |\xi|\langle \xi \rangle - v \cdot \xi | } \lesssim \begin{cases} 1 & | \xi | \leq 1 \text{,} \\ | \xi |^{-2} & | \xi | \geq 1 \text{,} \end{cases} \]
and hence is in $L^2 (\RR^3)$.
Thus, we can estimate
\[ \| U H_v^{-1} f \|_{ L^\infty } \lesssim \| m \hat{f} \|_{ L^1 } \lesssim \| m \|_{ L^2 } \| \hat{f} \|_{ L^2 } \lesssim \| f \|_{ L^2 } \text{.} \]
The symbols of the operators $H_v^{-1}$, $U H_v^{-1} H_{v^\prime}^{-1} \nabla$, and $H_v^{-1} H_{v^\prime}^{-1} \nabla$ are also in $L^2 (\RR^3)$, and the bounds \eqref{eq.trivial_unif_cor} follow.

Finally, \eqref{eq.trivial_comp_cor} follows from the above observations and the identity
\[ ( H_v^{-1} - H_{v^\prime}^{-1} ) f = i ( v - v^\prime ) \cdot H_v^{-1} H_{v^\prime}^{-1} \nabla f \text{.} \qedhere \]
\end{proof}

Let $G_v = U_r^{-1} \gamma_v$, with $v \in \R^3$ and $|v| < 1$, as in Section \ref{sec.inertpf_intro}.
For $v$ and $v^\prime$ as in Corollary \ref{thm.trivial_unif_cor}, \eqref{eq.trivial_comp_cor} implies that
\begin{equation} \label{eq.traveling_wave_comp} \| U ( G_v - G_{v^\prime} ) \|_{ L^\infty } + \| G_v - G_{v^\prime} \|_{ L^\infty } \lesssim_{ |v|, |v^\prime| } | v - v^\prime | \| W \|_{ L^2 } \text{.} \end{equation}

The next batch of estimates involves decay in time and forms the backbone of the proofs of \eqref{eq.field_decay_re} and \eqref{eq.field_decay_im}.
The main technical ingredient is the dispersive estimate given in Proposition \ref{thm.dispersive_unif}.

\begin{lemma} \label{thm.trivial_decay_cor}
Let $v \in \RR^3$, with $| v | < 1$.
Furthermore, let $t \in (0, \infty)$, and let $f \in \mc{S} (\RR^3)$.
\begin{itemize}
\item The following decay estimates hold:
\begin{align}
\label{eq.trivial_decay_cor} \| e^{-i t H} f \|_{ L^\infty } &\lesssim t^{-\frac{3}{2}} ( \| \nabla f \|_{ L^1 } + \| f \|_{ L^1 } ) \text{,} \\
\notag \| e^{-i t H} U f \|_{ L^\infty } &\lesssim t^{-\frac{3}{2}} \| \nabla f \|_{ L^1 } \text{,} \\
\notag \| e^{-i t H} U H_v^{-2} \nabla f \|_{ L^\infty } &\lesssim t^{-\frac{3}{2}} \| f \|_{ L^1 } \text{.}
\end{align}

\item The following decay estimates hold for some small $\varepsilon > 0$:
\begin{align}
\label{eq.trivial_decay_cor_ex} \| e^{-i t H} U^{-1} f \|_{ L^\infty } &\lesssim_\varepsilon t^{-1 + \varepsilon} ( \| \nabla f \|_{ L^1 } + \| f \|_{ L^1 } ) \text{,} \\
\notag \| e^{-i t H} f \|_{ L^\infty } &\lesssim_\varepsilon t^{-1 + \varepsilon} \| \nabla f \|_{ L^1 } \text{,} \\
\notag \| e^{-i t H} H_v^{-2} \nabla f \|_{ L^\infty } &\lesssim_\varepsilon t^{-1 + \varepsilon} \| f \|_{ L^1 } \text{.}
\end{align}
\end{itemize}
\end{lemma}
\begin{proof}
For the first and third estimates of \eqref{eq.trivial_decay_cor}, we apply the dispersive estimate \eqref{eq.dispersive_unif}, with the arguments
\begin{alignat*}{3}
M &= I \text{,} &\qquad (a^\prime, b^\prime, a, l) &= \left( 0, 0, -\frac{1}{2}, 1 \right) \text{,} \\
M &= U H_{v_s}^{-2} \nabla \text{,} &\qquad (a^\prime, b^\prime, a, l) &= \left( 0, -3, -\frac{1}{2}, 0 \right) \text{.}
\end{alignat*}
For the second part of \eqref{eq.trivial_decay_cor}, we apply \eqref{eq.dispersive_unif_ex} to $| \nabla | f$, with
\[ M = U | \nabla |^{-1} \text{,} \qquad (a^\prime, b^\prime, a, l) = \left( 0, -1, -\frac{1}{2}, 0 \right) \text{.} \]

The proof of \eqref{eq.trivial_decay_cor_ex} is analogous: to prove the first and third estimates we again apply \eqref{eq.dispersive_unif}, but with the following parameters,
\begin{alignat*}{3}
M &= U^{-1} \text{,} &\qquad (a^\prime, b^\prime, a, l) &= (-1, 0, -1 - \varepsilon, 1) \text{,} \\
M &= H_{v_s}^{-2} \nabla \text{,} &\qquad (a^\prime, b^\prime, a, l) &= (-1, -3, -1 - \varepsilon, 0) \text{.}
\end{alignat*}
The second part of \eqref{eq.trivial_decay_cor_ex} follows from an application of  \eqref{eq.dispersive_unif_ex} to $| \nabla | f$, with
\[ M = | \nabla |^{-1} \text{,} \qquad (a^\prime, b^\prime, a, l) = (-1, -1, -1 - \varepsilon, 0) \text{.} \qedhere \]
\end{proof}

\subsubsection{Decay of $\delta$}

Using the preliminary estimates just established, we are able to prove decay of $\delta_t$, as $t \nearrow \infty$.
Since $D = U_r^{-1} \delta$, we have the identities
\begin{equation} \label{eql.delta_D} \Re \delta + i U \Im \delta = U D \text{,} \qquad U^{-1} \Re \delta + i \Im \delta = D \text{,} \end{equation}
so that decay estimates on $U D_t$ and $D_t$ will suffice.

We begin by analyzing $U D_t$.
Using the Duhamel formula \eqref{eq.D_duhamel}, we see that
\begin{align*}
\| U D_t \|_{ L^\infty } &\lesssim \| e^{-i t H} ( \Re \delta_0 ) \|_{ L^\infty } + \| e^{-i t H} ( U \Im \delta_0 ) \|_{ L^\infty } \\
&\qquad + \int_0^\frac{t}{2} \| e^{ -i H (t - s) } U H_{v_s}^{-2} ( \dot{v}_s \cdot \nabla W ) \|_{ L^\infty } ds \\
&\qquad + \int_\frac{t}{2}^t \| e^{ -i H (t - s) } U H_{v_s}^{-2} ( \dot{v}_s \cdot \nabla W ) \|_{ L^\infty } ds \\
&= I_1 + I_2 + I_3 + I_4 \text{.}
\end{align*}
To bound $I_1$ and $I_2$, we apply the first and second parts of \eqref{eq.trivial_decay_cor} and find that
\begin{align*}
I_1 &\lesssim t^{-\frac{3}{2}} ( \| \nabla ( \Re \delta_0 ) \|_{ L^1 } + \| \Re \delta_0 \|_{ L^1 } ) \lesssim t^{-\frac{3}{2}} \text{,} \\
I_2 &\lesssim t^{-\frac{3}{2}} \| \nabla ( \Im \delta_0 ) \|_{ L^1 } \lesssim t^{-\frac{3}{2}} \text{.}
\end{align*}
Similarly, for $I_3$, we apply the third part of \eqref{eq.trivial_decay_cor}:
\[ I_3 \lesssim t^{-\frac{3}{2}} \| W \|_{ L^1 } \int_0^\frac{t}{2} | \dot{v}_s | ds \lesssim t^{-\frac{3}{2}} \int_0^\frac{t}{2} ( 1 + s )^{-4} ds \lesssim t^{-\frac{3}{2}} \text{.} \]
Finally, to bound $I_4$, we apply the second part of \eqref{eq.trivial_unif_cor}, along with the fact that $e^{-i H (t - s)}$ is a bounded operator on $L^2 (\R^3)$, to conclude that
\[ I_4 \lesssim \| W \|_{ L^2 } \int_\frac{t}{2}^t | \dot{v}_s | ds \lesssim t ( 1 + t )^{-4} \lesssim ( 1 + t )^{-3} \text{.} \]
Combining the estimates for $I_1$ through $I_4$, it follows that
\begin{equation} \label{eq.decay_delta_re} \| \real \delta_t \|_{ L^\infty } \lesssim \| U D_t \|_{ L^\infty } \lesssim t^{-\frac{3}{2}} \text{.} \end{equation}

The proof of decay of $D_t$ in $t$ is analogous.
Again, by \eqref{eq.D_duhamel}, we have that
\begin{align*}
\| D_t \|_{ L^\infty } &\lesssim \| e^{-i t H} U^{-1} ( \Re \delta_0 ) \|_{ L^\infty } + \| e^{-i t H} ( \Im \delta_0 ) \|_{ L^\infty } \\
&\qquad + \int_0^\frac{t}{2} \| e^{ -i H (t - s) } H_{v_s}^{-2} ( \dot{v}_s \cdot \nabla W ) \|_{ L^\infty } ds \\
&\qquad + \int_\frac{t}{2}^t \| e^{ -i H (t - s) } H_{v_s}^{-2} ( \dot{v}_s \cdot \nabla W ) \|_{ L^\infty } ds \\
&= J_1 + J_2 + J_3 + J_4 \text{.}
\end{align*}
We bound $J_1$, $J_2$, and $J_3$ using \eqref{eq.trivial_decay_cor_ex}:
\begin{align*}
J_1 &\lesssim t^{-1 + \varepsilon} ( \| \nabla ( \Re \delta_0 ) \|_{ L^1 } + \| \Re \delta_0 \|_{ L^1 } ) \lesssim t^{-1 + \varepsilon} \text{,} \\
J_2 &\lesssim t^{-1 + \varepsilon} \| \nabla ( \Im \delta_0 ) \|_{ L^1 } \lesssim t^{-1 + \varepsilon} \text{,} \\
J_3 &\lesssim t^{-1 + \varepsilon} \int_0^\frac{t}{2} | \dot{v}_s | ds \lesssim t^{-1 + \varepsilon} \text{.}
\end{align*}
For $J_4$, we apply \eqref{eq.trivial_unif_cor} (like for $I_4$) and obtain
\[ J_4 \lesssim \int_\frac{t}{2}^t | \dot{v}_s | ds \lesssim ( 1 + t )^{-3} \text{.} \]
Combining these bounds yields
\begin{equation} \label{eq.decay_delta_im} \| \imag \delta_t \|_{ L^\infty } \lesssim \| D_t \|_{ L^\infty } \lesssim t^{-1 + \varepsilon} \text{.} \end{equation}

Bounds \eqref{eq.decay_delta_re} and \eqref{eq.decay_delta_im} prove the desired decay for $\delta$.

\subsubsection{Completion of the Proof of Theorem \ref{thm:technical}}

In order to complete the proofs of \eqref{eq.field_decay_re} and \eqref{eq.field_decay_im}, we first control the intermediate quantity $\beta_t - \gamma_{v_\infty}^{X_t}$.
Recalling \eqref{eq.decay_delta_re} and applying \eqref{eq.traveling_wave_comp}, we obtain that
\begin{align}
\label{eq.decay_interm_re} \| \real ( \beta_t - \gamma_{v_\infty}^{X_t} ) \|_{ L^\infty } &\lesssim \| \real \delta_t \|_{ L^\infty } + \| \real U ( G_{v_t} - G_{v_\infty} ) \|_{ L^\infty } \\
\notag &\lesssim t^{-\frac{3}{2}} + | v_\infty - v_t | \| W \|_{ L^2 } \\
\notag &\lesssim t^{-\frac{3}{2}} + \int_t^\infty | \dot{v}_s | ds \\
\notag &\lesssim t^{-\frac{3}{2}} \text{.}
\end{align}
By similar reasoning, using \eqref{eq.traveling_wave_comp} and \eqref{eq.decay_delta_im}, we also have that
\begin{equation} \label{eq.decay_interm_im} \| \Im ( \beta_t - \gamma_{v_\infty}^{X_t} ) \|_{ L^\infty } \lesssim t^{-1 + \varepsilon} \text{.} \end{equation}

Finally, in order to establish \eqref{eq.field_decay_re}, we use the triangle inequality
\[ \| \real ( \beta_t - \gamma_{v_\infty}^{ \bar{X}_0 + v_\infty t } ) \|_{ L^\infty } \lesssim \| \real ( \beta_t - \gamma_{v_\infty}^{X_t} ) \|_{ L^\infty } + \| \real ( \gamma_{v_\infty}^{ \bar{X}_0 + v_\infty t } - \gamma_{v_\infty}^{X_t} ) \|_{ L^\infty } \text{.} \]
The first term on the right-hand side is controlled by using \eqref{eq.decay_interm_re}.
To bound the remaining term, we apply the first inequality in \eqref{eq.trivial_unif_cor} to $\nabla W$ and then use \eqref{eq:xtconv}:
\begin{align*}
\| \real ( \gamma_{v_\infty}^{ \bar{X}_0 + v_\infty t } - \gamma_{v_\infty}^{X_t} ) \|_{ L^\infty } &\lesssim \| U ( G_{v_\infty}^{ \bar{X}_0 + v_\infty t } - G_{v_\infty}^{X_t} ) \|_{ L^\infty } \\
&\lesssim | \bar{X}_0 + v_\infty t - X_t | \| U \nabla G_{v_\infty} \|_{ L^\infty } \\
&\lesssim \| U H_{v_\infty}^{-1} \nabla W \|_{ L^\infty } ( 1 + t )^{-2} \\
&\lesssim \| \nabla W \|_{ L^2 } ( 1 + t )^{-2} \text{.}
\end{align*}
These estimates result in the bound \eqref{eq.field_decay_re}.
The remaining bound \eqref{eq.field_decay_im} for the imaginary part of the field is proven in a very similar manner.

This completes our proof of Theorem \ref{thm:technical}.

\appendix

\section{Global Well-Posedness} \label{sec.gwp}

In this section, we establish the well-posedness properties of the system \eqref{eq.emodel}.
More specifically, here we will prove Theorem \ref{thm.emodel_gwp}.
The proof is an application of standard methods, in particular the contraction mapping theorem.
The process here is particularly straightforward, as this system is close to linear.

\subsection{Local Well-Posedness}

The first step in the proof of Theorem \ref{thm.emodel_gwp} is the following general local well-posedness result.

\begin{proposition} \label{thm.emodel_lwp}
Fix $a \in (-\infty, 1/2)$ and $b \in [1, \infty)$, let $P_0 \in \R^3$, and suppose both $\varphi_0 \in H_a^b (\R^3)$, $\pi_0 \in H_{a+1}^b (\R^3)$ are real-valued.
Then, there exists $T > 0$, depending on $\| W \|_{ H_{a+1}^b }$, $\| \varphi_0 \|_{ H_a^b }$, and $\| \pi_0 \|_{ H_{a+1}^b }$, such that the system \eqref{eq.emodel} has a unique solution
\[ X, P \in C ( [0, T]; \R^3 ) \text{,} \qquad (\real \beta, \imag \beta) \in C ( [0, T]; H_a^b (\R^3) \times H_{a+1}^b (\R^3) ) \text{,} \]
with corresponding initial data $X_0 = 0$, $P_0$, and $(\varphi_0, \pi_0)$.
\end{proposition}

\begin{proof}
For convenience, we let
\[ \mc{R} = \| W \|_{ H_{a+1}^b } \text{,} \qquad R = \| \varphi_0 \|_{ H_a^b } + \| \pi_0 \|_{ H_{a+1}^b } \text{,} \]
and we make the usual change of variables, $B = U_r^{-1} \beta$, as in Section \ref{sec.equations_cov}.
Integrating the second equation in \eqref{eq.emodel_alt} and applying \eqref{eq.B_duhamel}, we obtain the fixed point equation $\Phi (P) = P$, where the map $\Phi$, from $C ( [0, T]; \R^3 )$ into itself, is given by
\begin{align}
\label{eql.emodel_lwp_2} [ \Phi (Q) ]_t &= P_0 + \real \int_0^t \int_{ \R^3 } ( \nabla W )^{Y_s} e^{- \i H s} U B_0 dx ds \\
\notag &\qquad - \real \i \int_0^t \int_{ \R^3 } ( \nabla W )^{Y_s} \int_0^s e^{ -\i H (s - \tau) } U W^{Y_\tau} d\tau dx ds \text{,}
\end{align}
with $Y \in C ( [0, T]; \R^3 )$ defined
\[ Y_t = X_0 + \int_0^t Q_s ds = \int_0^t Q_s ds \text{.} \]
It suffices now to show $\Phi$ is a contraction on $C ( [0, T]; \R^3 )$ for small enough $T$.

Given $Q, Q^\prime \in C ( [0, T]; \R^3 )$, we have
\begin{align}
\label{eql.emodel_lwp_3} [ \Phi (Q) - \Phi (Q^\prime) ]_t &= \real \int_0^t \int_{ \R^3 } I_1 (s) \cdot e^{- \i H s} U B_0 dx ds \\
\notag &\qquad - \real \i \int_0^t \int_0^s \int_{ \R^3 } I_2 (s, \tau) \cdot e^{ -\i H (s - \tau) } U W dx d \tau ds \text{,} \\
\notag I_1 (s) &= ( \nabla W )^{ Y_s } - ( \nabla W )^{ Y^\prime_s } \\
\notag I_2 (s, \tau) &= ( \nabla W )^{ Y_s - Y_\tau } - ( \nabla W )^{ Y^\prime_s - Y^\prime_\tau } \text{,}
\end{align}
where $Y^\prime$ is constructed from $Q^\prime$ in the same manner as $Y$ from $Q$.
To bound \eqref{eql.emodel_lwp_3}, we begin by writing
\begin{align*}
\| \Phi (Q) - \Phi (Q^\prime) \|_{ L^\infty } &\leq T \sup_{ s \in [0, T] } \| I_1 (s) \|_{ H_{-a}^{-b} } \| e^{- \i H s} U B_0 \|_{ H_a^b } \\
&\qquad + T^2 \sup_{ s, \tau \in [0, T] } \| I_2 (s, \tau) \|_{ H_{-a}^{-b} } \| e^{- \i H (s - \tau)} U W \|_{ H_a^b } \\
&\leq T \| B_0 \|_{ H_{a+1}^b } \sup_{ s \in [0, T] } \| I_1 (s) \|_{ H_{-a}^{-b} } \\
&\qquad + T^2 \| W \|_{ H_{a+1}^b } \sup_{ s, \tau \in [0, T] } \| I_2 (s, \tau) \|_{ H_{-a}^{-b} } \\
&\lesssim T R \sup_{ s \in [0, T] } \| I_1 (s) \|_{ H_{-a}^{-b} } + T^2 \mc{R} \sup_{ s, \tau \in [0, T] } \| I_2 (s, \tau) \|_{ H_{-a}^{-b} } \text{.}
\end{align*}

As a result, our goal is to control $I_1 (s)$ and $I_2 (s, \tau)$.
By the mean value theorem and the translation-invariance of the $H_a^b$-norms, we obtain
\begin{align*}
\sup_{ s \in [0, T] } \| I_1 (s) \|_{ H_{-a}^{-b} } &\lesssim \| Y - Y^\prime \|_{ L^\infty } \| \nabla^2 W \|_{ H_{-a}^{-b} } \\
&\lesssim T \| W \|_{ H_{-a+2}^{-b+2} } \| Q - Q^\prime \|_{ L^\infty } \\
&\lesssim T \mc{R} \| Q - Q^\prime \|_{ L^\infty } \text{.}
\end{align*}
In the last step, we used the assumptions $b \geq 1$ and $a \leq 1/2$ to bound the $H_{-a+2}^{-b+2}$-norm by the $H_{a+1}^b$-norm.
Similarly, we can bound $I_2 (s, \tau)$ as follows:
\[ \sup_{ s, \tau \in [0, T] } \| I_2 (s, \tau) \|_{ H_{-a}^{-b} } \lesssim \| Y - Y^\prime \|_{ L^\infty } \| \nabla^2 W \|_{ H_{-a}^{-b} } \lesssim T \mc{R} \| Q - Q^\prime \|_{ L^\infty } \text{.} \]

Finally, combining all the above, we obtain
\[ \| \Phi (Q) - \Phi (Q^\prime) \|_{ L^\infty } \lesssim \mc{R} ( T^2 R + T^3 \mc{R} ) \| Q - Q^\prime \|_{ L^\infty } \text{.} \]
Therefore, if $T$ is sufficiently small, depending on $R$ and $\mc{R}$, then $\Phi$ is a contraction on $C ( [0, T]; \R^3 )$.
We set $P$ to be the unique fixed point of $\Phi$.

With $P$ in hand, we can now define $X$ by
\[ X_t = \int_0^t P_s ds \text{,} \]
and then $B$ using the Duhamel formula \eqref{eq.B_duhamel}.
In particular, since $W$ has finite $H_{a+1}^b$-norm, then \eqref{eq.B_duhamel} implies $B \in C ( [0, T]; H_{a+1}^b (\R^3) )$.
It is clear that $X$, $P$, and $B$ form the unique solution to \eqref{eq.emodel_alt} in the given spaces, with initial data $X_0 = 0$, $P_0$, and $B_0$, respectively.
Finally, defining $\beta_t = U_r B_t$ for each $t \in [0, T]$, then
\[ \real \beta \in C ( [0, T]; H_a^b (\R^3) ) \text{,} \qquad \imag \beta = C ( [0, T]; H_{a+1}^b (\R^3) ) \text{.} \]
Moreover, $X$, $P$, and $\beta$ must form the unique solution to \eqref{eq.emodel} in the desired spaces, with initial data $X_0 = 0$, $P_0$, and $(\varphi_0, \pi_0)$, respectively.
\end{proof}

\subsection{Global Well-Posedness}

We can now complete the proof of Theorem \ref{thm.emodel_gwp}.
First of all, since $W \in L^1 (\R^3) \cap H^b (\R^3)$ by the assumptions of Theorem \ref{thm.emodel_gwp}, then $W \in H_{a+1}^b (\R^n)$ for all $a \in (-5/2, 1/2)$.
Suppose we have a local solution of \eqref{eq.emodel},
\[ X, P \in C ( [0, T); \R^3 ) \text{,} \qquad (\real \beta, \imag \beta) \in C ( [0, T); H_a^b (\R^3) \times H_{a+1}^b (\R^3) ) \text{,} \]
for a fixed $a$ in this range, with $T < \infty$.
Since the time of existence from Proposition \ref{thm.emodel_lwp} depends only on the size of the initial data for $\beta$, then it suffices to show
\begin{equation} \label{eql.emodel_gwp} \sup_{0 \leq t < T} \| B_t \|_{ H_{a+1}^b (\R^3) } \leq C \text{,} \end{equation}
where $B_t = U_r^{-1} \beta_t$.
If \eqref{eql.emodel_gwp} holds, then one can apply Proposition \ref{thm.emodel_lwp} with initial data $(X_{T - \epsilon}, P_{T - \epsilon}, \beta_{T - \epsilon})$, with sufficiently small $\epsilon > 0$, in order to push the above solution beyond time $T$.
This implies that there cannot be a finite maximal time of existence and hence complete the proof.

To establish \eqref{eql.emodel_gwp}, we can simply use \eqref{eq.B_duhamel}:
\begin{align*}
\| B_t \|_{ H_{a+1}^b } &\leq \| e^{-\i t H} B_0 \|_{ H_{a+1}^b } + T \sup_{s \in [0, T)} \| e^{-\i (t - s) H} W^{X_s} \|_{ H_{a+1}^b } \\
&\lesssim \| B_0 \|_{ H_{a+1}^b } + T \| W \|_{ H_{a+1}^b } \text{.} \qedhere
\end{align*}
This concludes the proof of Theorem \ref{thm.emodel_gwp}.

\section{Decay Estimates} \label{sec.decay}

We collect in this appendix the decay estimates, in both time and space, that are used in the paper.
First, we give a brief overview of Littlewood-Paley theory, which will be used in the decay estimates of Sections \ref{sec.decay_unif} and \ref{sec.decay_spatial}.

\subsection{Littlewood-Paley Theory} \label{sec.LP}

Let $\psi \in \mc{S} (\R^3)$ be compactly supported in the annulus $1/2 \leq | \xi | \leq 2$.
For each $k \in \mathbb{Z}$, we let $\psi_k \in \mc{S} (\R^3)$ be its rescaling
\[ \psi_k (\xi) = \psi (2^{-k} \xi) \text{.} \]
In addition, we can choose $\psi$ such that
\[ \sum_{k \in \mathbb{Z}} \psi_k (\xi) = 1 \text{,} \qquad \xi \in \R^3 \setminus \{ 0 \} \text{.} \]
We can also assume $\psi$ is spherically symmetric.

\begin{remark}
Note that by scaling properties, $\psi_k$ is in the class $\mc{M}_0^0$.
\end{remark}

Recall that the usual (homogeneous) Littlewood-Paley operators are defined as
\[ P_k: \mc{S} (\R^3) \rightarrow \mc{S} (\R^3) \text{,} \qquad \widehat{P_k u} = \psi_k \hat{u} \text{,} \qquad k \in \ZZ \text{.} \]
In other words, $P_k u$ is a ``projection'' of $u$ to the frequency band $| \xi | \simeq 2^k$.
For convenience, we sometimes adopt the notation
\[ P_{\sim k} = P_{k-1} + P_k + P_{k+1} \text{,} \qquad k \in \ZZ \text{.} \]
Due to Fourier support considerations, we have $P_k = P_k P_{\sim k}$.

Next, we recall the Bernstein, or finite band, inequalities, cf. \cite{bah_che_dan:fa_pde, tao:disp_eq}, which use the operators $P_k$ to convert derivatives to constants, and vice versa.

\begin{proposition} \label{thm.bernstein}
Let $p \in [1, \infty]$, and let $f \in \mc{S} (\R^3)$.
\begin{itemize}
\item If $b \in \R$, then
\begin{equation} \label{eq.bernstein_deriv} \| P_k | \nabla |^b f \|_{ L^p } \lesssim 2^{bk} \| f \|_{ L^p } \text{.} \end{equation}

\item In addition,
\begin{equation} \label{eq.bernstein_grad} \| P_k \nabla f \|_{ L^p } \lesssim 2^k \| f \|_{ L^p } \text{,} \qquad \| P_k f \|_{ L^p } \lesssim 2^{-k} \| \nabla f \|_{ L^p } \text{.} \end{equation}
\end{itemize}
\end{proposition}

In this paper, we require more than the usual finite band estimates found in \eqref{eq.bernstein_deriv} and \eqref{eq.bernstein_grad}.
We need extensions of these estimates to other Fourier multiplier operators.
Examples of such operators include $U$ and powers of $H_v$, where $v \in \RR^3$.
The generalized estimate we need is given in the subsequent proposition.

\begin{proposition} \label{thm.bernstein_mult}
Let $a, b \in \RR$, let $p \in [1, \infty]$, and let $M$ be a Fourier multiplier operator, with symbol $m: \R^3 \setminus \{ 0 \} \rightarrow \C$ in the class $\mc{M}_a^b$.
Then, for any $f \in \mc{S} (\R^3)$,
\begin{equation} \label{eq.bernstein} \| P_k M f \|_{ L^p } \lesssim_m \begin{cases} 2^{bk} \| f \|_{ L^p } & k \geq 0 \text{,} \\ 2^{ak} \| f \|_{ L^p } & k \leq 0 \text{.} \end{cases} \end{equation}
\end{proposition}

\begin{proof}
The proof is a straightforward adaptation of standard Fourier multiplier estimates; see, e.g., \cite[Lemma 2.2]{bah_che_dan:fa_pde}.
\end{proof}

\subsection{Subsonic Decay Estimates} \label{sec.decay_radial}

First, we discuss a simpler class of stationary phase estimates, in which the radial derivative of the phase never vanishes.
In this case, we can always integrate by parts in the radial direction.

In particular, this situation holds when the phase is the symbol $h_v$ of $H_v$ and when $|v| < 1$; see \eqref{eq.phase_subsonic}.
We encounter precisely this situation in the proof of Theorem \ref{thm:technical}, in particular for controlling the acceleration of the tracer particle.

\begin{proposition} \label{thm.dispersive_subsonic}
Assume the following:
\begin{itemize}
\item Let $l \in \ZZ$ and $b \in \RR$, with $l > 0$ and $b \leq 0$.

\item Let $v \in \RR^3$, with $|v| < 1$.

\item Suppose $m: \R^3 \setminus \{ 0 \} \rightarrow \C$ is in the class $\mc{O}^b$.

\item Let $f, g \in \mc{S} (\R^3)$.
\end{itemize}
Then, for any $t \in (0, \infty)$, we have the estimate
\begin{align}
\label{eq.dispersive_subsonic} &\left| \int_{ \R^3 } e^{-i t h_v ( \xi ) } | \xi |^l m (\xi) \hat{f} (\xi) \bar{\hat{g}} (\xi) \d \xi \right| \lesssim_{l, b, |v|, m} t^{-l - 3} \| \langle x \rangle^{l + 3} f \|_{ L^2 } \| \langle x \rangle^{l + 3} g \|_{ L^2 } \text{.}
\end{align}
\end{proposition}

\begin{proof}
Throughout, we write $\mc{O}^r$ to denote any function in the class $\mc{O}^r$.
Expressing \eqref{eq.dispersive_subsonic} in polar coordinates and letting $F = \hat{f} \bar{\hat{g}}$, it suffices to derive
\[ \left| \int_{ \Sph^2 } \int_0^\infty e^{-i t h_{v, \theta} (\rho) } \rho^{l + 2} \mc{O}^b (\rho) F_\theta (\rho) \d \rho \d \theta \right| \lesssim t^{-l - 3} \| \langle x \rangle^{l + 3} f \|_{ L^2 } \| \langle x \rangle^{l + 3} g \|_{ L^2 } \text{.} \]
Let $I$ denote the left-hand side of the above inequality.

Recalling \eqref{eq.phase_subsonic}, we bound
\begin{align*}
I &= t^{-1} \left| \int_{ \Sph^2 } \int_0^\infty \rho^{l + 2} \mc{O}^b (\rho) \frac{ \frac{d}{d \rho} e^{-i t h_{v, \theta} (\rho) } }{ h_{v, \theta}^\prime (\rho) } F_\theta (\rho) d \rho d \theta \right| \\
&\lesssim t^{-1} \left| \int_{ \Sph^2 } \int_0^\infty \rho^{l + 2} \mc{O}^{b - 1} (\rho) \frac{d}{d \rho} e^{-i t h_{v, \theta} (\rho) } F_\theta (\rho) d \rho d \theta \right| \text{.}
\end{align*}
Integrating the above by parts yields
\[ I \lesssim t^{-1} \left| \int_{ \Sph^2 } \int_0^\infty e^{-i t h_{v, \theta} (\rho) } \mc{O}^{b-1} (\rho) [ \rho^{l + 1} F_\theta (\rho) + \rho^{l + 2} F_\theta^\prime (\rho) ] d \rho d \theta \right| \text{.} \]
Moreover, we can iterate this process until all positive powers of $\rho$ are exhausted:
\begin{align*}
I &\lesssim t^{-l-3} \int_{ \Sph^2 } \int_0^\infty | \mc{O}^{b - l - 3} (\rho) | | F_\theta (\rho) | d \rho d \theta + t^{-l-3} | F_\theta (0) | \\
&\qquad\qquad + t^{-l-3} \int_{ \Sph^2 } \int_0^\infty | \mc{O}^{b - l - 2} (\rho) | \sum_{k = 0}^{l + 2} \rho^k | F_\theta^{ (k+1) } (\rho) | d \rho d \theta \text{.}
\end{align*}
Note that in the final integration by parts, since no powers of $\rho$ remain, we pick up an extra boundary term at the origin.

Noting our assumptions for $l$ and $b$, we have
\begin{align*}
I &\lesssim t^{-l - 3} ( \| F \|_{ L^\infty } + \| \nabla F \|_{ L^\infty } + \| \nabla^2 F \|_{ L^\infty } ) \\
&\qquad + t^{-l - 3} \int_{\Sph^2} \int_0^\infty \mc{O}^{b - l - 2} (\rho) \rho^2 \sum_{ k = 0 }^l \rho^k | F_\theta^{ (k + 3) } (\rho) | d \rho d \theta \\
&\lesssim I_1 + I_2 \text{.}
\end{align*}
For $I_1$, we can trivially bound
\begin{align*}
I_1 &\lesssim t^{-l - 3} \sum_{ |\alpha| + |\beta| \leq 2 } \| \partial^\alpha \hat{f} \|_{ L^\infty } \| \partial^\beta \hat{g} \|_{ L^\infty } \lesssim \| \langle x \rangle^2 f \|_{ L^1 } \| \langle x \rangle^2 g \|_{ L^1 } \text{.}
\end{align*}
For $I_2$, we convert back to Cartesian coordinates to obtain
\begin{align*}
I_2 &\lesssim t^{-l - 3} \sum_{ |\alpha| + |\beta| \leq l + 3 } \| \partial^\alpha \hat{f} \|_{ L^2 } \| \partial^\beta \hat{g} \|_{ L^2 } \lesssim t^{-l - 3} \| \langle x \rangle^{l + 3} f \|_{ L^2 } \| \langle x \rangle^{l + 3} g \|_{ L^2 } \text{.}
\end{align*}
Combining the above completes the proof.
\end{proof}

\begin{proposition} \label{lem:decayestimates}
Let $v \in \RR^3$, with $|v| < 1$, and let $f \in \mc{S} (\R^3)$, with $\hat{f}$ an even function.
Then, for any $N \geq 3$ and $t \in (0, \infty)$, we have
\begin{align}
\label{eq:decayestimates} \Im \scalar{ \nabla f }{ U \e^{-\i t H_v} H_v^{-2} \nabla f } \lesssim_{|v|, f, N} (1 + t)^{-N} \| \langle x \rangle^N f \|_{ L^2 }^2 \text{.}
\end{align}
\end{proposition}

\begin{proof}
For $1 \leq i, j \leq 3$, let
\[ I_{i, j} = \Im \scalar{ \partial_i f }{ U \e^{-\i t H_v} H_v^{-2} \partial_j g } \text{.} \]
By H\"older's inequality, we can trivially estimate $I_{i, j} \lesssim \| f \|_{ L^2 }^2$.

Writing $I_{i, j}$ in polar coordinates, we have
\[ I_{i, j} = \Im \int_{\Sph^2} \theta_i \theta_j \int_0^\infty \frac{ \rho^3 | \hat{f} (\rho \theta) |^2 \e^{-\i t \rho ( \langle \rho \rangle - \theta \cdot v )} }{ \langle \rho \rangle ( \langle \rho \rangle - \theta \cdot v )^2 } \d \rho \d \theta =: \Im \int_{\Sph^2} \theta_i \theta_j I_\theta \d \theta \text{.} \]
Fixing $\theta \in \Sph^2$, we see there are smooth even functions
\[ u_\theta, w_\theta: \RR \rightarrow \RR \text{,} \qquad u_\theta = \frac{ | \hat{f} ( \rho \theta ) |^2 }{ \langle \rho \rangle ( \langle \rho \rangle - v \cdot \theta )^2 } \text{,} \qquad w_\theta (\rho) = \langle \rho \rangle - v \cdot \theta \text{,} \]
such that
\begin{align*}
I_\theta &= \int_0^\infty \rho^3 u_\theta (\rho) \e^{-\i t \rho w_\theta (\rho) } \d \rho \\
&= \frac{1}{2} \left( \int_0^\infty \rho^3 u_\theta (\rho) \e^{-\i t \rho w_\theta (\rho)} \d \rho - \int_0^\infty \rho^3 u_\theta (\rho) \e^{\i t \rho w_\theta (\rho)} \d \rho \right) \\
&= \frac{1}{2} \left( \int_0^\infty \rho^3 u_\theta (\rho) \e^{-\i t \rho w_\theta (\rho)} \d \rho + \int_{-\infty}^0 \rho^3 u_\theta (\rho) \e^{-\i t \rho w_\theta (\rho)} \d \rho \right) \text{,}
\end{align*}
where we applied a change of variables $\rho \rightarrow -\rho$ in the last step.
Thus,
\[ I_\theta = \frac{1}{2} \int_{-\infty}^\infty \rho^3 u_\theta (\rho) \e^{-\i t \rho ( \langle \rho \rangle - v \cdot \theta )} \d \rho \text{,} \]
and recalling the explicit form of $u_\theta$, we have
\[ I_{i, j} = \int_{ \Sph^2 } \theta_i \theta_j \int_{-\infty}^\infty \rho^3 \mc{O}^{-3} (\rho) | \hat{f} ( \rho \theta ) |^2 \e^{-\i t \rho ( \langle \rho \rangle - v \cdot \theta )} \d \rho \d \theta \text{.} \]

From here, the proof proceeds like for \eqref{eq.dispersive_subsonic}, except we no longer have a boundary at $\rho = 0$.
Because of \eqref{eq.phase_subsonic}, we can integrate by parts like in the proof of \eqref{eq.dispersive_subsonic} as many times as we wish.
Doing this $N$ times, we obtain
\begin{align*}
I_{i, j} &\lesssim t^{-N} \int_{\Sph^2} \left| \int_{-\infty}^\infty | \mc{O}^{-3 - N} (\rho) | [ \rho^3 \partial_\rho^N | \hat{f} (\rho \theta) |^2 + \rho^2 \partial_\rho^{N-1} | \hat{f} (\rho \theta) |^2 ] \d \rho \right| \d \theta \\
&\qquad + t^{-N} \int_{\Sph^2} \left| \int_{-\infty}^\infty | \mc{O}^{-3 - N} (\rho) | [ \rho \partial_\rho^{N-2} | \hat{f} (\rho \theta) |^2 + \partial_\rho^{N-3} | \hat{f} (\rho \theta) |^2 ] \d \rho \right| \d \theta \\
&\qquad + t^{-N} \int_{\Sph^2} \left| \int_{-\infty}^\infty \sum_{l = 0}^{N - 3} | \mc{O}^{2 N + l} (\rho) | \partial_\rho^{3 + l} | \hat{f} (\rho \theta ) |^2 \d \rho \right| \d \theta \text{.}
\end{align*}
The top two derivatives of $| \hat{f} |^2$ can be bounded in $L^2$, while the remaining lower derivatives can be bounded in $L^\infty$.
Thus, proceeding like in \eqref{eq.dispersive_subsonic}, we obtain
\begin{align*}
I_{i, j} &\lesssim t^{-N} \left( \| \langle x \rangle^N f \|_{ L^2 }^2 + \| \langle x \rangle^{N-1} f \|_{ L^2 }^2 + \sum_{l = 0}^{N-2} \| \langle x \rangle^l f \|_{ L^1 }^2 \right) \lesssim t^{-N} \| \langle x \rangle^N f \|_{ L^2 }^2 \text{.}
\end{align*}
This completes the proof of \eqref{eq:decayestimates}.
\end{proof}

\subsection{Dispersive Decay Estimates} \label{sec.decay_unif}

Next, we discuss general $L^\infty$-decay estimates for the operator $H$.
Since $H$ is of the same form as the linear differential operator studied in \cite{gustafson06}, we can take advantage of the decay estimate of \cite[Theorem 2.2]{gustafson06}.
The proof of this result relies on stationary phase techniques.

\begin{thm} \label{thm.dispersive_pre}
Let $h = h_0$, i.e., $h (\xi) = | \xi | \langle \xi \rangle$.
Then, for any $k \in \ZZ$, we have
\begin{align}
\label{eq.dispersive_pre} \sup_{ x \in \R^3 } \left| \int_{ \R^3 } \psi_k ( \xi ) e^{ -i h ( \xi ) t + i \xi \cdot x } d \xi \right| &\lesssim t^{ - \frac{3}{2} } 2^\frac{k}{2} \langle 2^k \rangle^{-\frac{1}{2}} \text{,} \\
\notag \sup_{ x \in \R^3 } \left| \int_{ \R^3 } \frac{ \psi_k ( \xi ) }{ h( \xi ) } e^{ -i h ( \xi ) t + i \xi \cdot x } d \xi \right| &\lesssim t^{ - \frac{1}{2} } 2^\frac{k}{2} \langle 2^k \rangle^{-\frac{1}{2}} \text{.}
\end{align}
\end{thm}
\begin{proof}
The first estimate follows directly from \cite[Theorem 2.2]{gustafson06}, with phase $h$ and amplitude $\psi_k$.
For the second estimate, we can write
\begin{align*}
\left| \int_{ \R^3 } \frac{ \psi_k ( \xi ) }{ h ( \xi ) } e^{ -i h ( \xi ) t + i \xi \cdot x } d \xi \right| &\lesssim \left| \int_t^\infty \int_{ \R^3 } \psi_k ( \xi ) e^{ -i h ( \xi ) s + i \xi \cdot x } d \xi ds \right| \\
&\lesssim 2^\frac{k}{2} \langle 2^k \rangle^{-\frac{1}{2}} \int_t^\infty s^{ - \frac{3}{2} } ds \text{,}
\end{align*}
where we applied the first estimate.
The desired bound follows.
\end{proof}

Theorem \ref{thm.dispersive_pre} implies the following dyadic dispersive estimate.

\begin{corollary} \label{thm.dispersive_dyadic}
If $f \in \mc{S} ( \R^3 )$, $t \in (0, \infty)$, and $0 \leq \sigma \leq 1$, then
\begin{equation} \label{eq.dispersive_dyadic} \| P_k H^{-\sigma} e^{-i t H} f \|_{ L^\infty_x } \lesssim t^{ - \frac{3}{2} + \sigma } \| P_{\sim k} U^\frac{1}{2} f \|_{ L^1_x } \text{,} \end{equation}
where $P_{\sim k} = P_{k-1} + P_k + P_{k+1}$.
\end{corollary}

\begin{proof}
We first consider the case $\sigma = 0$.
Consider the function
\[ g_k: \R^3 \rightarrow \C \text{,} \qquad g_k (x) = \int_{ \R^3 } \psi_k ( \xi ) e^{ -i h ( \xi ) t + i \xi \cdot x } d \xi \text{,} \]
where $h$ is as in the statement of Theorem \ref{thm.dispersive_pre}.
Then,
\[ \| P_k e^{-i t H} f \|_{ L^\infty_x } \lesssim \| g_k \ast P_{\sim k} f \|_{ L^\infty_x } \lesssim \| g_k \|_{ L^\infty_x } \| P_{\sim k} f \|_{ L^1_x } \text{.} \]
The $L^\infty$-bound for $g_k$ follows from the first estimate of \eqref{eq.dispersive_pre}, so that
\[ \| P_k e^{-i t H} f \|_{ L^\infty_x } \lesssim t^{ - \frac{3}{2} } 2^\frac{k}{2} \langle 2^k \rangle^{-\frac{1}{2}} \| P_{\sim k} f \|_{ L^1_x } \lesssim t^{ - \frac{3}{2} } \| P_{\sim k} U^\frac{1}{2} f \|_{ L^1_x } \text{,} \]
where in the last step, we applied \eqref{eq.bernstein} to the operator $U^{-1/2}$.

The case $\sigma = 1$ is proved similarly, using the second estimate in \eqref{eq.dispersive_pre}.
For the remaining $0 < \sigma < 1$, the result can be established by interpolation.
\end{proof}

Finally, we piece together the dyadic estimates \eqref{eq.dispersive_dyadic} in the various frequency bands into a single dispersive estimate.
This is stated below, in Proposition \ref{thm.dispersive_unif}, in a form applicable to a general class of Fourier multiplier operators.

\begin{proposition} \label{thm.dispersive_unif}
Let $a^\prime, b^\prime \in \RR$, and let $M$ be a Fourier multiplier operator, with symbol of the form $\mc{M}_{a^\prime}^{b^\prime}$.
Also, suppose $a \in \RR$ and $l \in \ZZ$ satisfy $-3/2 \leq a < a^\prime$, $l \geq 0$, and $b^\prime < l$.
Then, for any $f \in \mc{S} (\R^3)$,
\begin{align}
\label{eq.dispersive_unif} \| e^{-i t H} M f \|_{ L^\infty } &\lesssim_{m, a^\prime, b^\prime, a, l} t^{-\frac{3}{2}} \| \nabla^l f \|_{ L^1 } + t^{- \min ( \frac{3}{2}, 2 + a ) } \| f \|_{ L^1 } \text{,} \\
\label{eq.dispersive_unif_ex} \| e^{-i t H} M | \nabla | f \|_{ L^\infty } &\lesssim_{m, a^\prime, b^\prime, a, l} t^{-\frac{3}{2}} \| \nabla^{l+1} f \|_{ L^1 } + t^{- \min ( \frac{3}{2}, 2 + a ) } \| \nabla f \|_{ L^1 } \text{,}
\end{align}
\end{proposition}

\begin{proof}
First, we apply a Littlewood-Paley decomposition to obtain
\[ \| e^{-i t H} M f \|_{ L^\infty } \leq \sum_{k < 0} \| P_k e^{-i t H} M f \|_{ L^\infty } + \sum_{k \geq 0} \| P_k e^{-i t H} M f \|_{ L^\infty } = \mc{L} + \mc{H} \text{.} \]
For $\mc{H}$, letting $l = b^\prime + \varepsilon$ and applying \eqref{eq.bernstein_grad} and \eqref{eq.bernstein}, we obtain
\[ \mc{H} \lesssim \sum_{k \geq 0} 2^{-\varepsilon k} 2^{ ( b^\prime + \varepsilon ) k} \| P_k e^{-i t H} f \|_{ L^\infty } \lesssim \sum_{k \geq 0} 2^{-\varepsilon k} \| P_k e^{-i t H} \nabla^l f \|_{ L^\infty } \text{.} \]
We can now apply the dyadic dispersive estimate \eqref{eq.dispersive_dyadic} to the above.
Noting that $P_k U^{1/2}$ is bounded on $L^1$ by \eqref{eq.bernstein} and that $2^{-\varepsilon k}$ is summable yields
\[ \mc{H} \lesssim t^{-\frac{3}{2}} \sum_{k \geq 0} 2^{-\varepsilon k} \| P_{\sim k} U^\frac{1}{2} \nabla^l f \|_{ L^1 } \lesssim t^{-\frac{3}{2}} \| \nabla^l f \|_{ L^1 } \text{.} \]

Next, for the low-frequency terms $\mc{L}$, we write $a = a^\prime - \varepsilon$, and we apply \eqref{eq.bernstein} to $M$, $U^{1/2}$, and $H^{-a - 1/2}$.
This results in the estimate
\[ \mc{L} \lesssim \sum_{k < 0} 2^{\varepsilon k} 2^{a^\prime k - \varepsilon k} \| P_k e^{-i t H} f \|_{ L^\infty } \lesssim \sum_{k < 0} 2^{\varepsilon k} \| P_k e^{-i t H} H^{a + \frac{1}{2}} U^{-\frac{1}{2}} f \|_{ L^\infty } \text{.} \]
We can now apply the dyadic dispersive estimate \eqref{eq.dispersive_dyadic}.
If $a \leq -1/2$, then
\[ \mc{L} \lesssim t^{-2 - a} \sum_{k < 0} 2^{\varepsilon k} \| P_{\sim k} f \|_{ L^1 } \lesssim t^{-2 - a} \| f \|_{ L^1 } \text{.} \]
Otherwise, if $a > -1/2$, then
\[ \mc{L} \lesssim t^{-\frac{3}{2}} \sum_{k < 0} 2^{\varepsilon k} \| P_{\sim k} H^{a + \frac{1}{2}} f \|_{ L^1 } \lesssim t^{-\frac{3}{2}} \| f \|_{ L^1 } \text{,} \]
since positive powers of $H$ are bounded on low frequency bands by \eqref{eq.bernstein}.
Combining the above bounds for $\mc{H}$ and $\mc{L}$ yields \eqref{eq.dispersive_unif}.

Finally, for the remaining estimate \eqref{eq.dispersive_unif_ex}, we follow the same steps as above, but with $f$ replaced by $| \nabla | f$, in order to obtain
\begin{align*}
\| e^{-i t H} M | \nabla | f \|_{ L^\infty } &\lesssim t^{-\frac{3}{2}} \sum_{k \geq 0} 2^{-\varepsilon k} \| P_{\sim k} | \nabla | U^\frac{1}{2} \nabla^l f \|_{ L^1 } \\
&\qquad + t^{-\frac{3}{2}} \sum_{k < 0} 2^{\varepsilon k} \| P_{\sim k} | \nabla | H^{a + \frac{1}{2}} f \|_{ L^1 } \text{.}
\end{align*}
Using \eqref{eq.bernstein_deriv} and \eqref{eq.bernstein_grad}, in addition to \eqref{eq.bernstein} as before, yields \eqref{eq.dispersive_unif_ex}.
\end{proof}

\subsection{Spatial Decay Estimates} \label{sec.decay_spatial}

Finally, we obtain some estimates involving spatial decay of functions.
These are used in the proof of Theorem \ref{thm.traveling_wave_reg} in order to extract decay rates of subsonic and sonic inertial solutions.
Given a multi-index $\alpha = (\alpha_1, \alpha_2, \alpha_3)$ and any $x \in \R^3$, we will let $x^\alpha$ denote the polynomial $x_1^{\alpha_1} x_2^{\alpha_2} x_3^{\alpha_3}$.

The key component of the decay estimates in the subsonic setting is the following Fourier-localized uniform bounds.

\begin{lemma} \label{thm.spatial_decay_lem}
Let $a, b \in \RR$, and let $M$ be a Fourier multiplier operator, with symbol $m$ in the class $\mc{M}_a^b$.
Also, fix a multi-index $\alpha = (\alpha_1, \alpha_2, \alpha_3)$ and an integer $k \geq 0$.
Then, for any $f \in \mc{S} (\R^3)$ and $x \in \R^3$, we have
\begin{align}
\label{eq.spatial_decay_high} | x^\alpha P_k M f (x) | &\lesssim_{m, b} 2^{ ( 3/2 + b ) k } \| \langle x \rangle^{ |\alpha| } f \|_{ L^2 } \text{,} \\
\label{eq.spatial_decay_low} | x^\alpha P_{-k} M f (x) | &\lesssim_{m, a} 2^{ ( |\alpha| - a - 3 ) k } \| \langle x \rangle^{ |\alpha| } f \|_{ L^1 } \text{.}
\end{align}
\end{lemma}

\begin{proof}
First, for \eqref{eq.spatial_decay_high}, by the Fourier inversion formula, we can write
\begin{align*}
| x^\alpha P_k M f (x) | &\simeq \left| \int_{ \R^3 } e^{ i x \cdot \xi } \partial^\alpha [ \psi_k (\xi) m (\xi) \hat{f} (\xi) ] d \xi \right| \text{.}
\end{align*}
Since $\psi_k m$ is in the class $\mc{M}_a^b$ and is supported in the band $|\xi| \simeq 2^k$,
\begin{align*}
| x^\alpha P_k M f (x) | &\lesssim \sum_{\beta_1 + \beta_2 = \alpha} \int_{ \R^3 } | \partial^{\beta_1} \hat{f} (\xi) | | \partial^{\beta_2} [ \psi_k (\xi) m (\xi) ] | d \xi \\
&\lesssim \sum_{\beta_1 + \beta_2 = \alpha} 2^{ ( b - |\beta_2| ) k } \int_{ | \xi | \simeq 2^k } | \partial^{\beta_1} \hat{f} (\xi) | d \xi \text{.}
\end{align*}
By H\"older's inequality, we can control $\hat{f}$ in $L^2$:
\begin{align*}
| x^\alpha P_k M f (x) | &\lesssim 2^{ ( 3/2 + b ) k } \sum_{ |\beta| \leq |\alpha| } \| \partial^\beta \hat{f} \|_{ L^2 } \lesssim 2^{ ( 3/2 + b ) k } \| \langle x \rangle^{ |\alpha| } f \|_{ L^2 } \text{.}
\end{align*}

Next, for \eqref{eq.spatial_decay_low}, we first write
\begin{align*}
| x^\alpha P_{-k} M f (x) | &\simeq \left| \int_{ \R^3 } e^{i x \cdot \xi } \partial^\alpha [ \psi_{-k} (\xi) m (\xi) \hat{f} (\xi) ] d \xi \right| \text{.}
\end{align*}
Since $\psi_{-k} m$ is of the form $\mc{M}_a^b$,
\begin{align*}
| x^\alpha P_{-k} M f (x) | &\lesssim \sum_{ \beta_1 + \beta_2 = \alpha } \| \partial^{ \beta_1 } \hat{f} \|_{ L^\infty } 2^{ ( |\beta_2| - a ) k } \int_{ | \xi | \simeq 2^{-k} } d \xi \\
&\lesssim 2^{ ( |\alpha| - a - 3 ) k } \| \langle x \rangle^{ |\alpha| } f \|_{ L^1 } \text{.} \qedhere
\end{align*}
\end{proof}

We can construct a decay estimate by summing over the frequency-localized estimates of Lemma \ref{thm.spatial_decay_lem}.
This following estimate is applied in order to find the spatial decay rate for subsonic inertial solutions in Theorem \ref{thm.traveling_wave_reg}.

\begin{proposition} \label{thm.spatial_decay}
Let $a \in \RR$, and let $M$ be a Fourier multiplier operator, with symbol $m$ in the class $\mc{M}_a^{-2}$.
If $\mu \in [0, a + 3)$, then
\begin{equation} \label{eq.spatial_decay} | x |^\mu | M f (x) | \lesssim_{m, \mu, \mu^\prime} \| \langle x \rangle^{ \mu^\prime } f \|_{ L^2 } \text{,} \end{equation}
for any $x \in \R^3$, $f \in \mc{S} (\R^3)$, and $\mu^\prime > \lceil \mu \rceil + 3/2$.
\end{proposition}

\begin{proof}
We write $\mu = \mu_0 + \mu_1$, where $\mu_0 \in \ZZ$ and $\mu_1 \in [0, 1)$.
Applying a frequency decomposition to the left-hand side of \eqref{eq.spatial_decay} and then utilizing Lemma \ref{thm.spatial_decay_lem},
\begin{align*}
| x |^\mu | M f (x) | &\lesssim \sum_{k \geq 0} [ | x |^{ \mu_0 + 1 } | P_k M f (x) | ]^{\mu_1} [ | x |^{ \mu_0 } | P_k M f (x) | ]^{ 1 - \mu_1 } \\
&\qquad + \sum_{k > 0} [ | x |^{ \mu_0 + 1 } | P_{-k} M f (x) | ]^{\mu_1} [ | x |^{ \mu_0 } | P_{-k} M f (x) | ]^{ 1 - \mu_1 } \\
&\lesssim \sum_{k \geq 0} 2^{ -\frac{1}{2} k } \| \langle x \rangle^{ \lceil \mu \rceil } f \|_{ L^2 } + \sum_{k > 0} 2^{ ( \mu - a - 3 ) k } \| \langle x \rangle^{ \lceil \mu \rceil } f \|_{ L^1 } \text{.}
\end{align*}
In particular, we applied \eqref{eq.spatial_decay_high} and \eqref{eq.spatial_decay_low}, with $|\alpha| = \mu_0$ and $|\alpha| = \mu_0 + 1$.
Finally, since $\mu - a - 3 < 0$, then the above sums converge, and \eqref{eq.spatial_decay} follows.
\end{proof}

\begin{remark}
One can also prove Proposition \ref{thm.spatial_decay} by restricting to the Fourier domain $| \xi | \geq \delta$, deriving decay estimates independent of $\delta$ in this region, and letting $\delta \searrow 0$.
This has the added convenience of avoiding the full dyadic decomposition.
However, the interpolation between integer powers of decay becomes more complicated in this case.
\end{remark}

In the sonic setting in Theorem \ref{thm.traveling_wave_reg}, one requires more care in handling the low frequencies.
This is due to the more singular nature of the operator $H_v^{-1}$ near the Fourier origin when $|v| = 1$.
The following lemma provides the key component of this analysis in the low-frequency region.

For notational convenience, given a ``sonic velocity" $v \in \Sph^2$, we define $S_v$ to be the Fourier multiplier operator given by
\begin{equation} \label{eq.Sv} \widehat{( S_v f )} ( \xi ) = \frac{ \hat{f} (\xi) }{ \langle \xi \rangle - v \cdot \hat{\xi} } \text{.} \end{equation}
In particular, the denominator is especially singular as $|\xi| \rightarrow 0$ in the $v$-direction.

\begin{lemma} \label{thm.spatial_decay_lem_ex}
Let $a, b \in \RR$, and let $M$ be a Fourier multiplier operator, with symbol $m$ in the class $\mc{M}_a^b$.
Also, fix $v \in \Sph^2$, a multi-index $\alpha = (\alpha_1, \alpha_2, \alpha_3)$, and an integer $k \geq 0$.
Then, for any $f \in \mc{S} (\R^3)$ and $x \in \R^3$, we have
\begin{align}
\label{eq.spatial_decay_low_ex} | x^\alpha P_{-k} M S_v f (x) | &\lesssim \begin{cases} \| \langle x \rangle^{ |\alpha| } f \|_{ L^1 } 2^{ ( 3 |\alpha| - a - 3 ) k } & |\alpha| > 0 \text{,} \\ \| f \|_{ L^1 } (1 + k) 2^{ - ( a + 3 ) k } & |\alpha| = 0 \text{.} \end{cases}
\end{align}
\end{lemma}

\begin{proof}
As in the proof of \eqref{eq.spatial_decay_low}, we begin by writing
\begin{align*}
| x^\alpha P_{-k} M S_v f (x) | &\simeq \left| \int_{ \R^3 } e^{i x \cdot \xi} \partial^\alpha \left[ \frac{ \psi_{-k} (\xi) m (\xi) \hat{f} (\xi) }{ \langle \xi \rangle - v \cdot \hat{\xi} } \right] d \xi \right| \text{,}
\end{align*}
and we expand the right-hand side using the Leibniz rule.
Derivatives of $m$, $\psi_{-k}$, and $\hat{f}$ are handled in the same manner as in the proof of \eqref{eq.spatial_decay_low}.
For the extra factor, we note that if $l > 0$, then
\[ \nabla_\xi [ ( \langle \xi \rangle - v \cdot \hat{\xi} )^{-l} ] = g (\xi) \cdot ( \langle \xi \rangle - v \cdot \hat{\xi} )^{-l - 1} \text{,} \]
with $g$ in the class $\mc{M}^0_{-1}$.
Combining the above, we obtain
\begin{align*}
| x^\alpha P_{-k} M S_v f (x) | &\lesssim \sum_{ \beta_1 + \beta_2 + \beta_3 = \alpha } \| \partial^{ \beta_1 } \hat{f} \|_{ L^\infty } 2^{ ( |\beta_2| - a ) k } \int_{ | \xi | \simeq 2^{-k} } | \partial^{\beta_3} [ ( \langle \xi \rangle - v \cdot \hat{\xi} )^{-1} ] | d \xi \\
&\lesssim \sum_{ \beta_1 + \beta_2 + \beta_3 = \alpha } \| \partial^{ \beta_1 } \hat{f} \|_{ L^\infty } 2^{ ( |\beta_2| + | \beta_3 | - a ) k } \int_{ | \xi | \simeq 2^{-k} } \frac{ d \xi }{ ( \langle \xi \rangle - v \cdot \hat{\xi} )^{ 1 + |\beta_3| } } \\
&\lesssim \| \langle x \rangle^{ |\alpha| } f \|_{ L^1 } 2^{ ( |\alpha| - a ) k } \int_{ |\xi| \simeq 2^{-k} } \frac{ d \xi }{ ( \langle \xi \rangle - v \cdot \hat{\xi} )^{ 1 + |\alpha| } } \text{.}
\end{align*}

We must now control the integral on the right-hand side.
To do this, we convert to spherical coordinates $(\rho, \theta, \phi)$ like in the proof of \eqref{eq.traveling_wave_sonic}, with $\phi \in [0, \pi]$ measuring the angle from the $v$-direction.
This yields
\begin{align*}
\int_{ |\xi| \simeq 2^{-k} } \frac{ d \xi }{ ( \langle \xi \rangle - v \cdot \hat{\xi} )^{ 1 + |\alpha| } } &\lesssim \int_0^{2 \pi} \int_0^\pi \int_{ \rho \sim 2^{-k} } \frac{ \rho^2 \sin \phi }{ ( \langle \rho \rangle - \cos \phi )^{ 1 + |\alpha| } } d \rho d \phi d \theta \\
&\lesssim 2^{-3 k} \int_0^\pi \frac{ \phi }{ ( 2^{-2k} + \phi^2 )^{ 1 + |\alpha| } } d \phi \text{,}
\end{align*}
where in the last step, we used that $\sin \phi \lesssim \phi$, $1 - \cos \phi \simeq \phi^2$, and $\langle \rho \rangle - 1 \simeq \rho^2 \simeq 2^{-2k}$ in the domain of integration.
This last integral can be evaluated explicitly:
\begin{align*}
\int_0^\pi \frac{ \phi }{ ( 2^{-2k} + \phi^2 )^{ 1 + |\alpha| } } d \phi &\lesssim \begin{cases} ( 2^{-2k} + \phi^2 )^{ -|\alpha| } |_0^\pi \lesssim 2^{ 2 k |\alpha| } & |\alpha| > 0 \text{,} \\ \log ( 2^{-2k} + \phi^2 ) |_0^\pi \lesssim 1 + k & |\alpha| = 0 \text{.} \end{cases}
\end{align*}

Combining the above computations, we obtain \eqref{eq.spatial_decay_low_ex}.
\end{proof}

Finally, we state and prove the main estimate pertaining to sonic decay.

\begin{proposition} \label{thm.spatial_decay_ex}
Let $a \in \RR$, let $v \in \Sph^2$, and let $M$ be a Fourier multiplier operator, with symbol $m$ in the class $\mc{M}_a^{-1}$.
If $\mu \in [0, a/3 + 1)$, then
\begin{equation} \label{eq.spatial_decay_ex} | x |^\mu | M S_v f (x) | \lesssim_{M, \mu} \| \langle x \rangle^{ \mu^\prime } f \|_{ L^2 } \text{,} \end{equation}
for any $x \in \R^3$, $f \in \mc{S} (\R^3)$, and $\mu^\prime > \lceil \mu \rceil + 3/2$.
\end{proposition}

\begin{proof}
Again, we write $\mu = \mu_0 + \mu_1$, with $\mu_0 \in \ZZ$ and $\mu_1 \in [0, 1)$, so that
\begin{align*}
| x |^\mu | M S_v f (x) | &\lesssim \sum_{k \geq 0} [ | x |^{ \mu_0 + 1 } | P_k M S_v f (x) | ]^{\mu_1} [ | x |^{ \mu_0 } | P_k M S_v f (x) | ]^{ 1 - \mu_1 } \\
&\qquad + \sum_{k > 0} [ | x |^{ \mu_0 + 1 } | P_{-k} M S_v f (x) | ]^{\mu_1} [ | x |^{ \mu_0 } | P_{-k} M S_v f (x) | ]^{ 1 - \mu_1 } \\
&= H + L \text{.}
\end{align*}
For the high-frequency terms $H$, note that $M S_v$ is now in the class $\mc{M}_a^{-2}$, so that, like in the proof of \eqref{eq.spatial_decay}, the estimate \eqref{eq.spatial_decay_high} yields the desired bound for $H$.
For the low-frequency terms, we apply \eqref{eq.spatial_decay_low_ex}:
\begin{align*}
L &\lesssim \| \langle x \rangle^{ \lceil \mu \rceil } f \|_{ L^1 } \sum_{k > 0} ( 1 + k ) 2^{ [ 3 (\mu_0 + 1) - a - 3 ] \mu_1 k } 2^{ [ 3 \mu_0 - a - 3 ] (1 - \mu_1) k } \\
&\lesssim \| \langle x \rangle^{ \lceil \mu \rceil } f \|_{ L^1 } \sum_{k > 0} ( 1 + k ) 2^{ ( 3 \mu - a - 3 ) k } \text{.}
\end{align*}
Finally, since $3 \mu - a - 3 < 0$, then the above series converges, and
\[ L \lesssim \| \langle x \rangle^{ \mu^\prime } f \|_{ L^2 } \text{.} \qedhere \]
\end{proof}

\bibliographystyle{plain}

\end{document}